\documentclass[showpacs,amsmath,amssymb,notitlepage,hidelinks]{revtex4-1}

\usepackage{color}
\usepackage{graphicx}
\usepackage{amsmath}
\usepackage{amsfonts}
\usepackage{amsthm}
\usepackage{amscd}
\usepackage{epsfig}
\usepackage{amssymb}
\usepackage{tabularx}
\usepackage{calligra}
\usepackage{enumerate}
\usepackage{hyperref}
\usepackage{mathrsfs}

\usepackage{amsthm}
\usepackage{verbatim} 

\usepackage{soul}

\newtheorem{thm}{Theorem}
\newtheorem{prop}[thm]{Proposition}

\renewcommand{\a}{\alpha}
\renewcommand{\b}{\beta}
\newcommand{\g}{\gamma}
\newcommand{\e}{\epsilon}
\newcommand{\m}{\mu}

\renewcommand{\d}{\delta}

\newcommand{\w}{\omega}

\newcommand{\p}{\partial}

\def\c{\nabla}

\def\hd{\widehat{D}}

\begin{document}
\title{Instability of asymptotically anti de Sitter black holes under Robin conditions at the timelike boundary}
\author{Bernardo Araneda} \email{baraneda@famaf.unc.edu.ar}
\author{Gustavo Dotti} \email{gdotti@famaf.unc.edu.ar}
\affiliation{Facultad de Matem\'atica,
Astronom\'{\i}a y F\'{\i}sica, Universidad Nacional de C\'ordoba,\\
Instituto de F\'{\i}sica Enrique Gaviola, Conicet.\\
Ciudad Universitaria, (5000) C\'ordoba, Argentina}

\begin{abstract}

The static region outside the event horizon of an asymptotically  anti de Sitter black hole 
has a conformal timelike boundary $\mathscr{I}$ on which boundary conditions have to be imposed    
for the  evolution of linear fields from initial data to be 
a well posed problem. Only  homogeneous Dirichlet, Neumann or Robin
conditions preserve the action of the background isometry group  on the solution space. 
We study the case in which the modal decomposition of the linear field leads to potentials not diverging 
at the conformal timelike boundary.  
 We prove that  there is always an instability if Robin boundary conditions 
with large enough $\gamma$ (the quotient between the values of the derivative of the field and the field at the boundary) are allowed.  
We explain the origin of this instability, show that for  modes with nonnegative potentials  
there is a single unstable state and prove a number of properties of this state. 
Although our results apply in general 
  to 1+1 wave equations on a half infinite domain with a potential that is not singular at the boundary, 
 our  motivation  is 
 to analyze the  gravitational  stability of the four dimensional 
 Schwarzschild anti de Sitter black holes (SAdS${}_4$) 
 in the context of the 
 black hole non modal linear stability program 
initiated in  Phys.\ Rev.\ Lett.\  {\bf 112}, 191101 (2014), and the related 
supersymmetric type of duality exchanging odd and even modes. 
We prove that this symmetry is broken except when 
 a  combination of Dirichlet conditions in the even sector and a  particular  Robin condition in the odd sector is enforced,  
or viceversa, and that  only the first of these two choices leads to a stable dynamics. 
\end{abstract}

\maketitle

\tableofcontents

\section{Introduction}

A preliminary  stability criterion  for  a stationary black hole is that linear fields on the outer stationary region 
  do not grow unbounded from initial data. Fields of interest are Klein Gordon, Maxwell  and  linear perturbations of the 
metric. 
The evolution of  linearized metric perturbations is particularly important because it gives a hint  about 
 the ultimate  question of full non linear stability, which is  whether or not  generic initial data for the gravitational field 
equations, close to that of the black hole,  
 will evolve into  spacetimes that asymptotically approach  stationary black holes of similar characteristics. 
There are cases, however,  where the complexity of the metric and 
the field equations make even an integral treatment of the  linear gravity problem particularly complex  (some 
examples are   higher dimensional hairy black holes 
in generalized gravity theories).  In those cases 
 the stability of scalar and/or Maxwell fields 
is often considered as  indicative of linear gravity stability. \\
In any case, the stability notion  assumes  unique evolution from initial data, 
which, given that  the fields of interest  obey hyperbolic equations,  is guaranteed only if the outer region is globally hyperbolic. 
Asymptotically anti de Sitter spacetimes, however,  are not globally hyperbolic, they have a conformal timelike boundary $\mathscr{I}$ where  boundary 
 conditions have to be imposed to guarantee unambiguous evolution from initial data. 
When  different choices of boundary conditions are possible,  they lead  to 
different dynamics outside the domain of dependence of the initial data hypersurface, and therefore 
to potentially  different  answers to the issue of stability. \\

We are interested in 
the four dimensional Schwarzschild  black hole solution  of General Relativity  (GR) 
\begin{equation}  \label{metric}
ds^2 = - f(r) dt^2 + \frac{dr^2}{f(r)} + r^2 \left( d\theta^2+ \sin^2 \theta \; d\varphi^2 \right),
\end{equation}
when the cosmological constant $\Lambda$ in 
\begin{equation} \label{f4m}
f(r) = 1-2M/r -\Lambda r^2/3,
\end{equation}
is chosen to be negative ($M>0$ is the mass).  This is the Schwarzschild anti de Sitter black hole in four dimensions  (SAdS${}_4$), for which 
$f$ has a unique positive root at $r=r_h$, the horizon radius, in terms of which 
\begin{equation}\label{f4h}
f(r)= 1- \frac{(1-\tfrac{1}{3}\Lambda r_h^2) r_h}{r} - \frac{\Lambda r^2}{3}.
\end{equation}
In \cite{Dotti:2013uxa} a definite answer to the  linear stability problem for the Schwarzschild black hole 
  ($\Lambda=0$ in (\ref{f4m})) was given, by showing that  {\em generic metric perturbations} 
 will remain bounded and that the perturbed metric will approach asymptotically that of 
a slowly rotating Kerr black hole. All previous works on Schwarzschild linear stability had been  restricted to  isolated harmonic 
modes. The proof of nonmodal linear stability in \cite{Dotti:2013uxa} was  extended to the $\Lambda >0$ case in 
\cite{Dotti:2016cqy}. 
Crucial to this  proof  is the use of the 
duality between odd and even Schwarzschild perturbation 
discovered by Chandrasekhar \cite{ch1}. This duality  gives a bijection between 
even and odd parity modes \cite{Dotti:2013uxa} \cite{Dotti:2016cqy}, and 
this bijection  allows to write {\em even perturbations} in terms of 
solutions of the odd perturbation Regge-Wheeler equation in $(t,r)$ space (for a detailed proof see  Lemma 7 in \cite{Dotti:2016cqy}).
 The even mode Zerilli equation is therefore avoided 
  and 
 the linear stability problem for the 
  $\Lambda \geq 0$ Schwarzschild black hole shown  to reduce to the study of a {\em four dimensional} Regge-Wheeler 
equation for a {\em scalar} field $\Phi$
\begin{equation} \label{4drwe}
\nabla_{\a} \nabla^{\a} \Phi + \left( \frac{M}{r^3} - \frac{2 \Lambda}{3} \right) \Phi=0
\end{equation}
on the $\Lambda \geq 0$ background \cite{Dotti:2013uxa} \cite{Dotti:2016cqy}. \\

Two difficulties arise when trying to generalize these ideas to  SAdS${}_4$:
\begin{itemize}
\item one, of a fundamental nature, is that  
there are different dynamics depending on the boundary conditions at the timelike conformal boundary,
\item the other one
is that for most boundary conditions Chandrasekhar's duality, which is instrumental in the treatment 
of the $\Lambda \geq 0$ cases, is broken. \\
\end{itemize}

Introduce, 
 as usual, 
 a ``tortoise" coordinate $x$
\begin{equation} \label{x4}
x = - \int_r^{\infty} \frac{dr'}{f(r')}, 
\end{equation}
If $\Lambda$ were nonnegative we would find that $x \in (-\infty,\infty)$, however for negative $\Lambda$ 
we find that $x \in (-\infty,0)$ 
with $x \to -\infty$ as $r \to {r_h}^+$ and $x \to 0$ as $r \to \infty$ in the following way:
\begin{equation} \label{x42}
x  \simeq \begin{cases} \frac{r_h}{1-\Lambda {r_h}^2} \ln \left( \frac{r}{r_h}-1 \right) &, r \to r_h^+ \\
\frac{3}{\Lambda r} &, r \to \infty. 
\end{cases} 
\end{equation}
In  terms of $x$ the metric on the static region reads 
\begin{equation} \label{metric2}
ds^2 = f (-dt^2+dx^2) + r^2 (d\theta^2 + \sin^2\theta \; d\phi^2), 
\end{equation}
where $r=r(x)$ is the inverse of (\ref{x4}). This metric is defined on 
the  manifold    $\mathbb{R}_t \times (-\infty,0)_x \times S^2$ and has 
  the same causal structure of $S^2$ times  {\em the} $x<0$ 
{\em half} of Minkowski space in 1+1 dimensions, which is 
 not a  globally hyperbolic spacetime:  given 
 any spacelike hypersurface $\Sigma$  
there will be causal lines not intersecting it (e.g., those which are future directed and  originate at an $x=0$  point to the future of $\Sigma$).
Fields obeying wave like equations are no longer determined by 
their values and time derivatives at, say,  a $t=t_o$ surface $\Sigma_{t_o}$. The differential equation they satisfy 
has a unique solution only within the domain of dependence of $\Sigma_{t_o}$, and this 
is not the entire space. To assure uniqueness on the entire space, boundary conditions at the conformal timelike boundary $x=0$ have to be 
specified.  
Different  boundary conditions may be consistent with the field equations and yet 
lead to different evolutions of the same initial datum.\\

The situation of SAdS${}_4$ generalizes to the large class of asymptotically AdS static black hole solutions of  $d=n+2$ dimensional GR  
with horizon an Einstein manifold $\sigma^n$ with metric $\hat g_{AB}$ and Ricci tensor $\hat R_{AB}= (n-1) \kappa \hat g_{AB}$, $\kappa=0,
\pm1$. 
The metric of these black holes  in static coordinates $(t, r, z^A)$ is given by \cite{Gibbons:2002pq} \cite{Ishibashi:2011ws}
\begin{equation} \label{bhm}
ds^2 = -f(r) dt^2 + \frac{dr^2}{f(r)} + r^2 \hat g_{AB}(z) dz^A dz^B, 
\end{equation}
with 
\begin{equation} \label{f}
f(r) = \kappa - \frac{2M}{r^{n-1}}  - \frac{2\Lambda r^2}{n(n+1)},
\end{equation}
$M$ the mass and $\Lambda$ the cosmological constant. 
For $\Lambda<0, M > 0$ and any $\kappa$, $f$ grows monotonically for $r \in (0,\infty)$ from minus to plus infinity, with a simple  zero at $r=r_h$
 and 
\begin{equation} \label{fh}
f =  \kappa -  \frac{\left(\kappa {r_h}^{(n-1)} - \tfrac{2 \Lambda {r_h}^{n+1}}{n(n+1)}\right)}{r^{n+1}}
- \frac{2\Lambda r^2}{n(n+1)}.
\end{equation}
 The static region  corresponds to $r_h < r < \infty$, where we define $x$ as in (\ref{x4}) and find that 
\begin{equation} \label{x}
x  \simeq \begin{cases}\tfrac{1}{ f'(r_h)} \ln \left( \frac{r}{r_h}-1 \right) &, r \to r_h^+ \\
\frac{n(n+1)}{2 \Lambda r} &, r \to \infty. 
\end{cases} 
\end{equation}
In terms of $x$ the static region metric is 
\begin{equation} \label{metric-n}
ds^2 = f (-dt^2+dx^2) + r^2  \hat g_{AB}(z)\; dz^A dz^B.
\end{equation}
As for SAdS${}_4$, $x \in (-\infty,0)$,  
then the  static region  manifold    $\mathbb{R}_t \times (-\infty,0)_x \times \sigma^n$ has
the causal structure  of $\sigma^n$ times {\em a half} of  1+1 Minkowski spacetime.\\  

Alternative theories of gravity, such as Lovelock's, admit asymptotically (A)dS black hole solutions 
with metrics of the form (\ref{bhm}) (see \cite{Garraffo:2008hu} and references therein). 
The function  $f$ is no longer given by (\ref{f}), but in the asymptotically AdS case, by definition,  $f \sim -\Lambda_{eff}
 r^2$ for 
large $r$ and some negative effective cosmological 
constant $\Lambda_{eff}$, whereas near the event horizon $r=r_h$ (the largest positive root of $f$) $f \sim f'(r_h)(r-r_h)$, 
thus, 
the integral (\ref{x4}) that defines $x$ (so that (\ref{metric-n}) holds)  diverges logarithmically as $r \to r_h^+$ 
and converges for $r \to \infty$, then 
$x$ 
is again restricted to a half line. This is very different to what happens in the asymptotically flat  case, where 
$f \sim 1$ for large $r$, then $x \sim r$  in this limit and $x \in (-\infty,\infty)$. Similarly, in the asymptotically de Sitter case, 
the static region corresponds to  $r_h<r<r_c$ ($r_c$ is the cosmological horizon). Here  $r_h$ and $r_c$ 
are simple roots of $f$ that introduce  logarithmic divergences near both horizons in 
the integral defining $x$, therefore $x \in (-\infty, \infty)$. 
The metric can always be put in the form (\ref{metric-n}), but only in 
the asymptotically AdS case is $x$  restricted to a half line and the static region fails to be  globally hyperbolic. \\

For metrics  of the form  (\ref{bhm}) with constant curvature horizons $\sigma^n$, scalar and  Maxwell fields, as well as 
linear  metric perturbations, can all be expanded 
 as a   series in a basis of eigentensors of the Laplace-Beltrami (LB) operator 
on $\sigma^n$,  with ``coefficients'' that carry tensor  indexes in the $(t,r)$ Lorentzian {\em orbit manifold} 
\cite{Kodama:2003ck,Ishibashi:2011ws}. We call each term in this series a field {\em mode}. 
After some work, the (Maxwell, linear gravity, etc) 
field equations reduce in every case  to an infinite set of 1+1 wave equations for a master variable, 
 one for each mode, with 
a time independent potential. Since the massless wave equation in 1+1 dimension is conformally invariant, 
the master equation satisfied by the master variable has the Minkowskian   form 
\begin{equation}\label{2dwe}
\left[-\frac{\p^2}{\p t^2}   + \frac{\p^2}{\p x^2} - V(x) \right] \phi = 0, \;\;\; x<0.
\end{equation}
For GR in arbitrary dimensions, this reduction was proved 
  by Kodama and Ishibashi \cite{Kodama:2003ck,Ishibashi:2011ws}. For black holes with constant curvature horizons in 
the restricted case of second order Lovelock theories known as Einstein-Gauss-Bonnet 
gravity, the modal reduction of linear gravity  to the form (\ref{2dwe}) 
was done in \cite{Dotti:2005sq,Gleiser:2005ra}.  Generalizations to higher order 
Lovelock theories can be found in \cite{Takahashi:2010ye}. Further examples of reduction of linear field equations  to mode equations of 
the form (\ref{2dwe}) include  hairy black holes, as in \cite{Anabalon:2015vda}. In all these cases, it is the 
 warped product form of the background (\ref{bhm}) 
what allows separation of variables, independently of the form of $f$. 
The horizon manifold drops out from the field equations, leaving 
a trace of it in the mode counting (modes are in one-one relation with  the eigenspaces of the LB operator on different kinds of tensor 
fields on $\sigma^{n}$), and on  the form of the potentials $V$ for each mode. 
The non globally hyperbolic character of 
the spacetime is what implies the existence of a conformal {\em timelike} boundary at $x=0$: 
contrary to what happens for $\Lambda \geq 0$, for 
which the domain of the wave equations (\ref{2dwe}) 
is 1+1 Minkowski spacetime, in the asymptotically AdS case $x<0$.\\
For perturbations propagating in the inner region of a Reissner-Nordstr\"om black hole or 
on the negative mass Schwarzschild or super extreme 
 Reissner-Nordstr\"om nakedly
 singular spaces, the modal decomposition leads again to the 1+1 wave equation 
(\ref{2dwe}) on a half-space. However, in these cases, due to the background curvature singularity, there is a unique self-consistent 
choice of boundary condition, and thus no ambiguity at all in the dynamics. \\

Since $V$ in equation (\ref{2dwe}) does not depend on $t$, we can solve it 
  by separating the $t-$variable, after which the problem  reduces to 
 finding the eigenvalues and eigenfunctions of a quantum Hamiltonian on a half-line:
\begin{equation} \label{qh}
\mathcal{H} = - \p^2/\p x^2 + V(x), \;\;\; x<0.
\end{equation}
In  this paper we  study the self adjoint extensions of the operator defined in  (\ref{qh}) and the corresponding solutions 
of equation (\ref{2dwe}). We do this for the case in which $V$ is continuous for $x \in (-\infty,0]$ since this case covers the gravitational 
perturbations of SAdS${}_4$ we are interested in. 
 We focus  on the different dynamics that arise under homogeneous Dirichlet:
 $$\phi|_{x=0}=0,$$ 
 Neumann 
$$\p_x \phi|_{x=0}=0$$
 or Robin 
$$\p_x \phi|_{x=0}=\g \phi|_{x=0}$$
 boundary conditions, which 
are the natural ones, as they preserve the action of spacetime isometries on the solution space (see \cite{Ishibashi:2003jd}),  
with an accent on the least studied Robin boundary condition, for which 
we find that there is always a range of the Robin parameter $\g$ under which the dynamics is unstable.  These different choices 
give all possible self-adjoint extensions of (\ref{qh}). \\

It is a non trivial fact that, as far as we know,  for scalar and  Maxwell fields and also for 
linearized gravity in four and higher dimensional GR and Lovelock theories with a timelike boundary,  $V \to 0$ as $x \to -\infty$ and 
either $V$ is continuous for $x \in (-\infty,0]$, or 
 $V$ is continuous in  $(-\infty,0)$ and diverges  at the conformal boundary {\em always} 
 as  $V \sim c/x^2$ for some constant $c$ that depends on the field type and the mode. 
For those   potentials   with $c < 3/4$ some generalized forms of Robin, Dirichlet and Neumann
 homogeneous boundary conditions 
can be defined, whereas if $c \geq 3/4$  only homogeneous  Dirichlet boundary conditions are allowed 
\cite{Ishibashi:2004wx}.  \\

The paper is organized as follows: Section \ref{dnghs} contains some preliminary  material 
on modal decomposition of linear fields and self adjoint extensions of operators like (\ref{qh}). 
Our results are gathered in Sections  \ref{robinst-section} and  \ref{robin4}. 
In Section \ref{robinst-section} we prove that for nonsingular potentials the 1+1 wave equation   has a critical value of the 
Robin parameter above which the system is unstable. We show that a 
boundary condition 
allowing energy flow from the boundary is not a sufficient condition for instability,  and  that 
the instability is due to the excitation of 
a unique negative energy bound state of the associated quantum Hamiltonian, of which  
we  give a number of properties.   These results are  collected in Propositions 1-6  and 
illustrated with two toy models. 
Although motivated by the study of linear gravity on AdS${}_4$, these results  apply to any problem reducing to equation 
(\ref{2dwe}) on a half space with a nonsingular potential, as well as to the related quantum mechanical problem 
on a half line.\\

Section \ref{robin4} 
explores the effects of the Robin instability of Maxwell fields and gravitational perturbations on SAdS${}_4$ with focus on 
how the boundary conditions spoil the even-odd duality, which is peculiar of four dimensions and is key in proving 
 the nonmodal stability of the Schwarzschild and Schwarzschild
de Sitter black holes. 
 We show in Proposition 7 that this duality, which is reminiscent of supersymmetric Quantum Mechanics, 
is broken except for two out of the infinitely many choices of boundary conditions, and that only one of these 
gives  a stable dynamics.\\

Except for the self-contained section on the scalar field on SAdS${}_4$, linear fields in asymptotically AdS backgrounds leading 
to singular potentials in their modal decompositions are not treated in this work. The massive scalar field 
on a warped background (\ref{bhm})  leads to singular potentials after mode decomposition. 
For a comprehensive nonmodal 
study of the massive 
scalar field on a {\em general} (i.e., non necessarily of the form (\ref{bhm})) asymptotically AdS background, the reader is addressed to 
references \cite{Holzegel:2012wt} \cite{Warnick:2012fi},\cite{Holzegel:2011uu}, \cite{Holzegel:2009ye} and \cite{Holzegel:2013kna}.

\section{Dynamics in a nonglobally hyperbolic background} \label{dnghs}

The subject of dynamics in a non globally hyperbolic static background has been treated in 
the series of papers \cite{wald}, \cite{Ishibashi:2003jd}, and \cite{Ishibashi:2004wx}. 
This is reviewed in  Section \ref{Robins}, where 
 we explain the relation between boundary conditions for the generalized eigenfunctions of (\ref{qh}) at $x=0$ 
(i.e., the different self-adjoint extensions of the operator $\mathcal{H}$) and the 
corresponding  dynamics driven by (\ref{2dwe}), which is the  equation obeyed by a linear field mode. 
The modal decomposition of linear fields on the  static warped backgrounds (\ref{bhm}) in arbitrary dimensions 
is treated in \cite{Ishibashi:2011ws} \cite{Kodama:2003ck} and references therein, 
\cite{Chaverra:2012bh} offers a detailed description of the four dimensional case with spherical symmetry. 
Linear fields are decomposed in independent modes and a master variable 
is obtained for each mode which satisfies a 1+1 wave equation of the form (\ref{2dwe}). 
In Section  \ref{sads4} we explain how this is done for linear fields on SAdS${}_4$.

\subsection{Boundary conditions and self-adjoint extensions} \label{Robins}

After expanding in modes a linear field, 
the field equations reduce to a set of equations, one for each mode, of the form (see (\ref{2dwe}) and (\ref{qh}))
\begin{equation}\label{dd22}
-\ddot \phi =  \mathcal{H} \; \phi, 
\end{equation}
where a dot means time derivative and   the Hamiltonian operator is 
\begin{equation} \label{hlm2}
\mathcal{H} = - \frac{\p^2}{\p x^2} + V(x),
\end{equation}
with $V$ continuous in $(-\infty,0)$ and  $V(x) \to 0$ as $x \to -\infty$. \\

The general solution of (\ref{dd22}) is of the form
\begin{equation}\label{sol11}
\phi = \int dE \;  c_E(t)\;  \psi_E(x), \;\; \; \ddot c_E + E c_E =0,\;\;\; 
\end{equation}
where the $\psi_E$ are generalized eigenfunctions
\begin{equation}\label{h}
{}^z\mathcal{H} \psi_E = E \; \psi_E
\end{equation}
of a chosen self adjoint extension ${}^z\mathcal{H}$ with domain a linear subset of $L^2((-\infty,0),dx)$. 
In what follows we will use $\mathcal{H}$ to denote the operator (\ref{hlm2}) acting on unspecified functions, 
and denote particular  self adjoint extensions to specified domains in $L^2((-\infty,0),dx)$ with a left upper superscript. 
Note that, whenever $\mathcal{H}$ admits different self adjoint extensions, the resulting dynamics from initial 
compact supported data (which is suitable for all extensions) will depend on the extension we choose.\\
According to (\ref{sol11}), a
  negative energy $E<0$ in the spectrum of the chosen extension ${}^z\mathcal{H}$ 
allows  exponentially growing terms in (\ref{sol11}). 
Thus, there will be an instability whenever the spectrum of ${}^z\mathcal{H}$ contains a negative $E$ value.
The possible self adjoint extensions depend on the behavior of $V$ near $x=0$ and minus infinity. 
The potential $V$ is said to be  {\em limit circle case} (LC) at $x=0$ if any function in the two dimensional space of 
 local solutions of the differential equation (\ref{h})  is square integrable near zero, and is  otherwise 
 said to be  {\em limit point case} (LP) at $x=0$. The same notion applies at minus infinity (see the appendix to section X.I in \cite{rs}).  
It is a non trivial fact that, according to Theorem X.6 b in \cite{rs}, the LC/LP notion does not depend on the value of $E$ in (\ref{h}) 
as long as $V$ is continuous for $x \in (-\infty,0)$. \\

The potentials that appear in the modal decompositions of linear fields are always LP at $x=-\infty$. Some of them are LC 
at $x=0$, and some others are LP at $x=0$. In the first case, {\em any} solution of (\ref{h}) is of the form 
$\psi=A \psi_1 + B \psi_2$, where both $\psi_1$ and $\psi_2$ are square integrable near $x=0$. 
 Writing $(A,B)=C (\cos(\a), \sin(\a))$ and allowing 
the irrelevant  overall factor $C$ 
 to be positive or negative while restricting   $\a \in (-\pi/2,\pi/2]$, we find that $\a$ parametrizes the set of allowed boundary conditions at $x=0$.
 If ${}^{\a}\mathcal{H}$ is the operator 
(\ref{hlm2}) with domain the linear subset of $L^2((-\infty,0),dx)$ of functions with boundary condition at $x=0$ consistent with choosing  $\a$ above,
then  ${}^{\a}\mathcal{H}$ is self adjoint. Any self adjoint extension of $\mathcal{H}$ is of this form. Thus, if $V$ is LC at $x=0$, 
 there are  infinitely many 
self adjoint extensions, and the dynamics is unambiguous only after selecting a self adjoint extension $\a$. On the other hand, 
if $V$ is LP at $x=0$ there is 
no such ambiguity, as there is a single choice of boundary condition under which $\mathcal{H}$ is self adjoint. In this case 
we have unique evolution of initial data at a $t$ slice in spite of the non globally hyperbolic character of 
the outer static region. This is ultimately due to the fact that we have constrained boundary conditions at 
$x=0$ to those making  $\mathcal{H}$ self-adjoint. In principle, there is a much wider set of possible 
boundary conditions. However,  if we require that the time translation symmetry acts in the expected way on the space of solutions 
and that there is a conserved energy, we are forced to adopt solutions of the form (\ref{sol11}) for some 
self adjoint extension of  $\mathcal{H}$ \cite{Ishibashi:2003jd}. Thus, the possibilities for linear fields evolving 
on the outer static region of AdS black holes are that either there is a unique dynamics (corresponding to 
the case where each $V$ is LP at $x=0$) or a set of dynamics 
parametrized by $\a \in (-\pi/2,\pi/2]$ for each mode such that  $V$ is LC at $x=0$.\\

 For massless scalar, Maxwell and gravitational perturbations  on SAdS${}_4$, 
 the behavior of the mode potentials near the horizon is  
\begin{equation} \label{vrh}
V = f(r) [U_o+ \mathcal{O}(r-r_h)], \;\;\; U_o \neq 0, \;\;\; r \simeq r_h.
\end{equation}
In this limit 
\begin{equation} \label{xath}
x \simeq \frac{\ln \left(\frac{r}{r_h}-1 \right)}{f'(r_h)}, 
\end{equation}
therefore (\ref{vrh}) gives an exponential decay of $V$ as $x \to -\infty$ (note that $f'(r_h)>0$):
\begin{equation}
V \simeq f'(r_h) \; r_h \; e^{f'(r_h) \; x}\;  [ U_o + ... ]
\end{equation}
The local Frobenius series solution of the differential equation (\ref{h}) near $r=r_h$ is, using (\ref{vrh}) 
\begin{equation} \label{minf}
\psi_E = C \;(r-r_h)^{\frac{\sqrt{-E}}{f'(r_h)}}\; \left[ 1 +\mathcal{O}(r-r_h) \right] +
 D\;  (r-r_h)^{-\frac{\sqrt{-E}}{f'(r_h)}}\;\left[ 1 + \mathcal{O}(r-r_h) \right].
\end{equation}
From (\ref{xath}), 
\begin{equation} \label{xath2}
(r-r_h)^{\pm \frac{\sqrt{-E}}{f'(r_h)}} \simeq  {r_h}^\frac{\pm\sqrt{-E}}{f'(r_h)} \;\;e^{\pm \sqrt{-E}\;  x},
\end{equation}
therefore  for positive $E$ any solution  (\ref{minf}) is oscillatory, which (as for Quantum Mechanics wave functions) 
is a suitable behavior for 
a generalized eigenfunction, even though it is  not square integrable near minus infinity under $dx$. 
For  negative  $E$  the exponentially growing solution has to be discarded, leaving only 
the $D=0$ solution (\ref{minf}), which decays   exponentially with $x$ as $x \to -\infty$. 
For no value of $E$ is {\em every} function  in the two dimensional space of local solutions of the differential
equation (18)  square integrable near $x=-\infty$: for negative $E$ the subspace of square integrable local solutions 
is one dimensional and for positive $E$ is trivial (as explained above, since $V$ is continuous for $x \in (-\infty,0)$, 
there is no need to check all cases: Theorem X.6 b in \cite{rs} 
guarantees that if the space of local square integrable solutions is less than two dimensional for one value of $E$, it will be so
for any $E$ in the complex plane).  
Thus, according to the definition above, $V$ is LP at $x=-\infty$.
For fields on higher dimensional GR or Lovelock backgrounds we find a similar pattern near the horizon. \\

In the limit  $r \to \infty$  we find  that for Maxwell and gravitational fields on SAdS${}_4$ 
 $V$ is continuous for $x \in (-\infty,0]$, with 
\begin{equation} \label{vc}
V(x) =  v_0 + v_1 x + v_2 x^2 + ..., \;\;\; ( v_{0} \neq 0 ), 
\end{equation}
for $x \simeq 0$. 
The modes of a massless scalar field on SAdS${}_4$ instead, diverges near $x=0$ as $V \sim 2/x^2$. \\
For fields on higher dimensional GR or Lovelock backgrounds we find similarly that either $V$ is nonsingular at 
$x=0$ or diverges as $V \sim c/x^2$. \\

When analyzing the local solutions of  (\ref{h}) near $x=0$ for both the nonsingular and the $V \sim c/x^2$ potentials, we find that 
 $E$ appears only  in sub-leading terms. Thus, 
any statement of local integrability near $x=0$ is independent of $E$, as anticipated by  Theorem X.6 b in \cite{rs}.

\subsection{Mode decomposition of linear fields on SAdS${}_4$} \label{sads4}

In this Section we illustrate  the modal decomposition for the Klein Gordon and Maxwell fields  and for linear metric perturbations   on the 
SAdS${}_4$ background (see, e.g.,  \cite{Chaverra:2012bh}).   
For a systematic treatment of the modal decomposition on dimension $D \geq 4$ black holes 
with constant curvature horizon manifolds we refer the reader to \cite{Ishibashi:2011ws,Kodama:2003ck}. 
As explained in the paragraph above equation (\ref{2dwe}), the word 
 {\em mode} in this context applies to the individual  terms  when writing  scalar and tensor fields   as
a series in a basis of eigen-scalar/tensors of the LB operator  of the constant curvature horizon manifold. 
In the case of SAdS${}_4$, this manifold is the unit sphere and 
 a {\em mode} corresponds to a fixed  $(\ell,m)$-harmonic 
solution of the field equations. Modes can be further decomposed into their time Fourier components 
$(\ell,m,\w)$ of frequency $\w$.\\ 

In all cases the modal decomposition reduces the problem of linear fields propagating on a background 
(\ref{metric-n}) to a set of 
1+1 wave equations of the form 
(\ref{2dwe}), independently of the dimension of the background.\\

\subsubsection{Massless scalar fields on SAdS${}_4$} \label{sf4d}
For the  SAdS${}_4$ metric (\ref{metric})-(\ref{f4m}),  the massless scalar field equation reads
\begin{equation} \label{we3}
0= \nabla_{\a} \nabla^{\a} \Phi   = (rf )^{-1} \; \left[ -\p_t^2+ f\p_{r} (f\p_{r})+ f \left( r^{-2} \triangle -f'/r\right)
\right] (r \Phi),
\end{equation}
where $\triangle$ is the Laplacian on the unit sphere
\begin{equation}
\triangle = \p_{\theta}^2 + \cot(\theta)\;  \p_{\theta} + \sin(\theta)^{-2}\;  \p_{\varphi}^2.
\end{equation}
Introducing 
\begin{equation}
\phi = r \Phi,
\end{equation}
and expanding  $\phi$ in an $L^2(S^2)$ orthonormal  basis $S_{(\ell,m)}$ of real spherical harmonics on the sphere, 
\begin{equation} \label{sca-exp}
\phi = \sum_{(\ell,m)}  \phi_{(\ell,m)}(t,x) \; S_{(\ell,m)}(\theta,\phi),
\end{equation}
we find   that  (\ref{we3}) is equivalent to a set of  wave equations of the form (\ref{dd22}) 
in the $x<0$ half of 1+1 Minkowski spacetime:
\begin{equation}\label{dd2}
-\ddot \phi_{(\ell,m)} =  \mathcal{H}_{\ell}^s \; \phi_{(\ell,m)},
\end{equation}
with Hamiltonian 
\begin{equation} \label{hlm}
\mathcal{H}_{\ell}^s = - \frac{\p^2}{\p x^2} + V_{\ell}^s,
\end{equation}
and potential 
\begin{equation} \label{VKG}
V_{\ell}^s = f \left[ \frac{2M}{r^3} + \frac{\ell (\ell+1)}{r^2} - \frac{2}{3} \Lambda  \right].
\end{equation}
Note that $V_{\ell}^s>0$ and that 
\begin{equation} \label{vslim}
V_{\ell}^s \sim \begin{cases} \frac{2}{x^2} - \frac{\Lambda}{3} (\ell^2+\ell+2) + \mathcal{O}(x) &,\text{as } x \to 0^- \\
(\ell(\ell +1)+1-\Lambda r_h^2) (1-\Lambda r_h^2) \; r_h^{-2} \;  \exp \left( \left(\frac{1-\Lambda r_h^2}{r_h} \right)x \right)   &, \text{as } x \to -\infty\end{cases}
\end{equation}

According to Theorem X.10 in \cite{rs}, the fact that the $x^{-2}$ coefficient in (\ref{vslim}) is greater than $3/4$ implies 
that the potential $V_{\ell}^s$ belongs to the   {\em limit point case} at $x=0$ (see Section \ref{Robins}), which means 
that the space of local solutions of 
\begin{equation} \label{see}
(-\p_x^2 + V_{\ell}^s) \psi = E \psi
\end{equation}
 for which $\int_c^0 \psi^2 dx < \infty$ for some negative $c$ is 
one dimensional. To check this note that the two dimensional space of Frobenius 
series solutions for (\ref{see}) at $x=0$ ($r = \infty$)  is
\begin{equation} \label{sx0}
\psi = A \left( \frac{1}{r^2} - \frac{3}{10 \Lambda^2}\frac{(\ell+3)(\ell-2) \Lambda + 3E}{r^4} + \mathcal{O}(r^{-5}) \right) 
+ B \left( r + \frac{3}{2 \Lambda^2} \frac{\ell(\ell+1) \Lambda+3E}{r} + \mathcal{O}(r^{-2}) \right), 
\end{equation}
and $B=0$ is required for the integral  $\int_{r_o}^{\infty} \psi^2 dr/f$ to converge. 
Since there is a unique  admissible condition 
at $x=0$, namely, $B=0$ in (\ref{sx0})  (i.e., Dirichlet), 
we conclude that the dynamics of massless scalar fields on SAdS${}_4$
 is not ambiguous. The extension of the domain of $\mathcal{H}^s_{\ell}=-\p_x ^2 + V^s_{\ell}$ 
from functions of compact support 
to allow $\psi \sim r^{-2}$ for large  $r$ ($B=0$ in (\ref{sx0})) 
gives a self adjoint operator  which, as we will now show, has a 
positive spectrum.\\

For  $x \to -\infty$ ($r \to r_h^+$), 
\begin{equation} \label{xrh}
\psi = C  (r-r_h)^{\kappa}\;  (1 + \mathcal{O}(r-r_h))  +  D  (r-r_h)^{-\kappa}\;  (1 + \mathcal{O}(r-r_h))
, \;\;\; \kappa = \frac{\sqrt{-E \; r_h^2}}{1-\Lambda r_h^2},
\end{equation}
(compare with (\ref{minf}))  where it is assumed that we take real part. \\
 If  $E>0$, a solution behaving like (\ref{sx0}) with  $B=0$ at infinity  will behave near the horizon as in (\ref{xrh}). 
This  oscillatory behavior is characteristic of generalized eigenfunctions 
for  potentials that vanish in this limit, equation   (\ref{vslim}).\\
However, a negative value of  $E$ (real positive $\kappa$ in (\ref{xrh}))
 would be  admissible only if there were solutions of (\ref{see}) behaving 
as in  (\ref{sx0}) with $B=0$  near $x=0$ {\em and} as in  (\ref{xrh}) with   $D=0$
near the horizon. 
To show that this is not possible, note that 
  assuming   $B=D=0$, and $E<0$,  $\int_{-\infty}^0 \psi^2 dx$ converges and we arrive at a contradiction:
\begin{multline}\label{pe}
E \int_{-\infty}^0 \psi^2 dx = \int_{-\infty}^0 \psi (-\p_x^2 \psi + V^s_{\ell}\psi) \; dx =\\\left[ \psi \; f\p_r \psi \right] \big|^{r=\infty}_{r=r_h} + 
 \int_{-\infty}^0 \left( (\p_x \psi)^2 + V_{\ell}^s \psi^2 \right) \; dx =  \int_{-\infty}^0 \left( (\p_x \psi)^2 + V_{\ell}^s \psi^2 \right) \; dx >0. 
\end{multline}

We conclude that the dynamics of massless scalar fields on SAdS${}_4$ is free of ambiguities, as there 
is a unique possible self adjoint extension ${}^D\mathcal{H}^s_{\ell}$
 (for Dirichlet) of $\mathcal{H}^s_{\ell}$ (that allowing only $B=0$ in (\ref{sx0})),
 and that the field has no modal 
instabilities, as for every $\ell$ the spectrum of ${}^D\mathcal{H}^s_{\ell}$ has no negative eigenvalues. \\
  Stronger results on  the stability of scalar fields on 
asymptotically AdS backgrounds can be found in 
\cite{Holzegel:2012wt,Warnick:2012fi,Holzegel:2011uu, 
Holzegel:2009ye,Holzegel:2013kna}, where   boundedness and decay was studied using 
a nonmodal approach 
  to the Klein Gordon  equation. \\

Maxwell fields, as well as gravitational perturbations, exhibit more complicated patterns:
the effective potentials of the modes are nonsingular at $x=0$ and this allows  for  infinitely many 
 boundary conditions, some of them leading to unstable dynamics. To show this, we need recall how to  
reduce Maxwell equations and the linearized Einstein equations (LEE) to a set of the form (\ref{dd22})-(\ref{hlm2}) using 
tensor decompositions into harmonic $S^2$ tensors,  
generalizing what was done   in (\ref{sca-exp}) for the scalar field. This is a well known procedure which, for $D=4$, is 
reviewed, e.g.,  in \cite{Dotti:2016cqy},
 \cite{Chaverra:2012bh} 
 and specifically for SAdS${}_4$  in \cite{Cardoso:2001bb}.

\subsubsection{Maxwell fields  on SAdS${}_4$} \label{smax4}

Write the Maxwell potential as a sum of its  vector/odd $(-)$ and scalar/even $(+)$ pieces, 
 $A_{\b}=A_{\b}^{(-)}+A_{\b}^{(+)}$ \cite{Dotti:2016cqy,Chaverra:2012bh,Cardoso:2001bb}, 
these can be gauge fixed to the form 
\begin{align} \label{maxo}
A_{\b}^{(-)} &= \sum_{(\ell,m)}  \phi^{(-,\ell,m)} \; (0, 0,  \tfrac{1}{\sin \theta} \; \p_{\phi} S_{(\ell,m)}, -  \sin \theta \;  \p_{\theta} S_{(\ell,m)})\\
A_{\b}^{(+)} &= \sum_{(\ell,m)} (f \p_r  \phi^{(+,\ell,m)} S_{(\ell,m)}, f^{-1} \p_t  \phi^{(+,\ell,m)} S_{(\ell,m)}, 0, 0). \label{maxe}
\end{align}
where $\phi^{(\pm,\ell,m)}$ are functions of $(t,r)$. Maxwell equations  $F=dA$ and  $\nabla^{\a} F_{\a \b}=0$ are equivalent to 
\begin{equation}
-\ddot \phi^{\pm}_{(\ell,m)} =  \mathcal{H}_{\ell}^{Max} \phi^{\pm}_{(\ell,m)},  
\end{equation}
where the Hamiltonian
\begin{equation}\label{ham}
 \mathcal{H}_{\ell}^{Max} =
-\p_x^2 + V_{\ell}^{Max}
\end{equation}
is independent of the $\pm$ parity and the azimuthal number $m$ and has a potential 
\begin{equation} \label{max4}
 V_{\ell}^{Max} = f \; \frac{\ell (\ell+1)}{r^2}.
\end{equation}
Note that $V_{\ell}^{Max}>0$ and that 
\begin{equation} \label{maxlim}
V_{\ell}^{Max} \sim \begin{cases}  - \frac{\Lambda}{3} \ell(\ell+1) + \mathcal{O}(x^2) &,\text{as } x \to 0^- \\
 \ell(\ell+1) (1-\Lambda r_h^2) \; r_h^{-2} \; \exp \left( \left(\frac{1-\Lambda r_h^2}{r_h} \right)x \right)  &, \text{as } x \to -\infty\end{cases}
\end{equation}
For future reference we also note that
\begin{equation} \label{vmi}
\int_{-\infty}^0 V_{\ell}^{Max}\; dx = \int_{r_h}^{\infty} V_{\ell}^{Max}\; \frac{dr}{f} = \frac{\ell (\ell+1)}{r_h}.
\end{equation}

\subsubsection{Gravitational waves on SAdS${}_4$} \label{grav4}

As done for  Maxwell fields, we can decompose 
 metric perturbations into scalar/even/+  and vector/odd/-  fields with harmonic numbers $(\ell,m)$. This procedure 
is well known, the details can be found in \cite{Chaverra:2012bh,Dotti:2016cqy}  and references therein. The $\ell=0,1$ sectors contain 
either pure gauge fields, or time independent fields which corresponds to perturbations within the Kerr family, that is, a 
variation of mass, or an addition of angular momentum.  These 
are irrelevant to the stability problem, so we will focus on the $\ell \geq 2$ modes. \\

For $\pm, \ell \geq 2$ modes the LEE reduce to the well known 
Regge-Wheeler $(-)$ and Zerilli $(+)$ equations for the gauge invariant fields $\phi^{\pm}_{(\ell,m)}(t,r)$, which 
 are  of the form (\ref{dd22})-(\ref{hlm2}). 
The Regge-Wheeler potential of the odd  Hamiltonian $\mathcal{H}^-_{\ell}$ is 
\begin{equation} \label{2drwv}
V^{(-)}_{\ell} = f \left( \frac{\ell (\ell+1)}{r^2} - \frac{6M}{r^3} \right),
\end{equation}
and  is positive in the outer static region $r>r_h$ 
except when  $r_h$ is small compared to $M$, case in which is negative in the interval $r_h < r < 6M/(\ell (\ell+1))$.
This potential behaves as 
\begin{equation} \label{gravlim}
 V_{\ell}^{(-)} \sim \begin{cases}  - \Lambda \ell(\ell+1)/3 + \mathcal{O}(x) &,\text{as } x \to 0^- \\
 (\ell(\ell+1)r_h-6M)(1-\Lambda r_h^2) r_h^{-3} \;  \exp \left( \left(\frac{1-\Lambda r_h^2}{r_h} \right)x \right)   &, \text{as } x \to -\infty\end{cases}
\end{equation}

The Zerilli potential of the even  Hamiltonian $\mathcal{H}^+_{\ell}$ is  
 \begin{equation} \label{zp}
V_{\ell}^{(+)}=f  \frac{[ \mu^2 \ell (\ell+1)-24M^2 \Lambda] r^3 + 6 \mu^2 M r^2 +36 \mu M^2 r +72M^3}{r^3 \; 
(6M+ \mu^2 r)^2}, \;\; \mu=(\ell-1)(\ell+2),
\end{equation}
this potential is positive for $r>r_h$ and  behaves as
\begin{equation} \label{gravlim}
 V_{\ell}^{(+)} \sim \begin{cases}  (24 M^2\Lambda^2) \mu^{-2}- \Lambda \ell(\ell+1) + \mathcal{O}(x) &,\text{as } x \to 0^- \\
 \frac{(\Lambda^2 r_h^4-4\Lambda r_h^2+\ell^4+2\ell^3-\ell^2-2\ell+3)(1-\Lambda r_h^2)}{(1-\Lambda r_h^2+\ell(\ell+1)) \;r_h^2} \;  \exp \left( \left(\frac{1-\Lambda r_h^2}{r_h} \right)x \right)   &, \text{as } x \to -\infty\end{cases}
\end{equation}
 For future reference we note that
\begin{equation} \label{vzi}
\int_{-\infty}^0 V_{\ell}^{(+)}\; dx = \frac{2 \Lambda^2 \; {r_h}^3}{ 3 (\ell+2)(\ell-1)} + 
\frac{\Lambda \; r_h  (\ell+3)(\ell-2)}{2(\ell+2)(\ell-1)} + \frac{2\ell^2+2 \ell-3}{2 r_h}.
\end{equation}

\section{Nonsingular potentials and Robin instabilities} \label{robinst-section}

By a nonsingular potential $V$ in (\ref{hlm2}) we mean one that is  
 continuous for $x \in (-\infty,0]$, satisfies (\ref{vrh})  and, for $x \simeq 0$, 
\begin{equation} \label{vc}
V(x) =  v_0 + v_1 x + v_2 x^2 + ..., \;\;\; ( v_{0} \neq 0 ).
\end{equation}
In this case, the local solutions near $x=0$ of the differential equation (\ref{h}) are of the form 
\begin{equation} \label{lst1}
\psi_E = K \cos (\alpha) \left[ 1 + \left(\frac{v_o-E}{2}\right)  x^2 + \frac{v_1}{6} x^3 + \mathcal{O}(x^4) \right] +
 K \sin (\alpha) \left[ x - \left(\frac{v_o-E}{6}\right) x^3 + \mathcal{O}(x^4) \right],
\end{equation}
where $A = K \cos(\alpha)=\psi_E(0)$ and $B= K \sin(\alpha)=\psi_E'(0)$,  and we choose to allow 
negative values of $K$ and restrict  $\alpha \in (-\pi/2,\pi/2]$. As discussed above (see the paragraph 
starting at equation (\ref{vrh})) the potential is LP at $x=-\infty$. From (\ref{lst1}) follows that it is LC at 
$x=0$, with boundary conditions parametrized by $\a$. Note that $\psi_E(0)$ and $\p_x\psi_E(0)$ are well defined 
and that the possible boundary conditions at $x=0$ are 
 Dirichlet ($\alpha=\pi/2$), Neumann ($\alpha= 0$) or Robin (the remaining cases). 
Robin boundary conditions  are characterized  by the nonzero value of 
\begin{equation} \label{gdef}
\psi_E'(0)/\psi_E(0)=\tan(\alpha) \equiv \gamma.
\end{equation}
Restricting to eigenfunctions satisfying (\ref{lst1}) with a fixed $\a$ value defines the spectrum 
of a self-adjoint extension of $\mathcal{H}$. 
We 
denote the corresponding self adjoint Hamiltonian operator for $\alpha=\pi/2$, $\a=0$ 
and $\g=\tan(\a) \neq 0,  \infty$ respectively as ${}^D\mathcal{H}$, ${}^N\mathcal{H}$ and 
${}^{\g}\mathcal{H}$.\\

 In this Section we study the possibility that a chosen self adjoint extension ${}^z\mathcal{H}$ admits a negative energy 
eigenfunction $\psi_E$, $E<0$, and establish a number of properties for such a state. 
The results of this Section apply then  to the  Quantum Mechanics problem on a half line with a non singular 
potential (as defined above) as well as  any 1+1 wave equation (\ref{dd22})-(\ref{hlm2}) with such a potential, for 
which, 
as explained above, a negative energy eigenvalue implies, in virtue of equation (\ref{sol11}),  that a generic 
solution of the 1+1 wave equation grows exponentially in time. Our motivation is, of course, spotting such 
instabilities in Maxwell fields and linearized gravity in SAdS${}_4$, whose modal decomposition 
lead to the nonsingular potentials (\ref{max4}), (\ref{2drwv}) and (\ref{zp}). \\

The asymptotic behavior near the horizon of a negative energy eigenfunction $\psi_E$, 
\begin{equation} \label{ef}
- \psi_E'' + V \psi_E =E \psi_E, \;\;\text{ for }  x \in (-\infty,0) , \;\; E<0 
\end{equation}
is that in (\ref{minf}) with $D=0$, then $\psi_E$ is  a ``bound state'',  that is, 
 belongs to $L^2((-\infty,0),dx)$ (as  follows from (\ref{minf}) and (\ref{xath2})):
\begin{equation} \label{bss}
 \int_{-\infty}^0 \psi_E^2 \; dx < \infty.
\end{equation}
If the chosen self adjoint extension is not Dirichlet, we may  normalize $\psi_E$ such that 
\begin{equation} \label{bcE}
 \psi_E(0) =1, \;\; \psi_E'(0) =\gamma.
\end{equation}

A key observation is that 
our intuition from Quantum Mechanics in $\mathbb{R}$ fails for the Schr\"odinger operator $\mathcal{H}$ on a half line 
subject to Robin boundary conditions, as this 
 may admit negative eigenvalues even if $V \geq 0$. The reason is that the ``kinetic energy'' operator $-\p_x^2$ 
fails to be positive definite if $\g >0$, as the following simple calculation of the expectation values of $\mathcal{H}$ for a 
(non necessarily normalized) real, square integrable  function on $(-\infty,0)$ shows: 
\begin{equation} \label{parts}
\left<\psi, \mathcal{H} \psi \right> :=  \int_{-\infty}^0 \psi \mathcal{H} \psi  \; dx = -\psi \psi' \big|_{x=0} + \int_{-\infty}^0  [(\psi')^2 + V \psi^2] \; dx.
\end{equation}
For Dirichlet or Neumann  boundary conditions at $x=0$ the first term on the right vanishes 
and we find that $V \geq 0$ does imply  
$\left<\psi, \mathcal{H} \psi \right> >0$, i.e., 
the self-adjoint extensions ${}^D\mathcal{H}$ and  ${}^N\mathcal{H}$ have  positive spectra. 
For a self-adjoint extension  ${}^{\gamma}\mathcal{H}$ corresponding to the Robin boundary condition 
\begin{equation} \label{robin}
\psi'(0) = \gamma \psi(0), \;\; \gamma \neq 0,
\end{equation}
we find from (\ref{parts}) that 
\begin{equation} \label{negstab}
\left<\psi, {}^{\g}\mathcal{H} \psi \right> = -\gamma \; \psi(0)^2 + \int_{-\infty}^0  [(\psi')^2 + V \psi^2] \; dx,
\end{equation}
which, assuming $V\geq 0$,  is positive for $\gamma<0$, but may be negative if $\gamma >0$. \\

\subsection{General results for nonsingular potentials} \label{grt1p}

In what follows, 
we use the fact that  any function in the domain of a self adjoint operator such as ${}^{\gamma} \mathcal{H}$ can be expanded 
in a basis of generalized eigenfunctions of the operator. 
Therefore, as in ordinary Quantum Mechanics, there is a function $\psi$ in this domain 
for which $\left< \psi,  {}^{\gamma} \mathcal{H} \psi \right> < 0$ if and only if the spectrum of  ${}^{\gamma} \mathcal{H}$ 
contains a negative value of $E$.

\begin{prop}[Robin instabilities]\label{robinst}
Let $V$ be nonsingular.  For large enough positive $\g$ the spectrum of ${}^{\gamma} \mathcal{H}$ contains a negative eigenvalue.
\end{prop}
\begin{proof}
Since nonsingular potentials are bounded, $|V(x)|  \leq V_o$ for some $V_o \geq 0$. Consider now 
the trial function $\psi=e^{\g x}$. If $\g>0$, this function belongs to the domain of ${}^{\gamma} \mathcal{H}$, and the expectation value 
of ${}^{\g}\mathcal{H}$ is 
\begin{equation}
\left< \psi , {}^{\g}\mathcal{H} \psi \right> = \int_{-\infty}^0 e^{\g x} \left( -\g^2 e^{\g x} + V e^{\g x}\right) dx 
= -\frac{\g}{2} + \int_{-\infty}^0  V e^{2\g x} \;dx  \leq -\frac{\g}{2} + \frac{V_o}{2 \g},
\end{equation}
which  is negative for 
\begin{equation} \label{c1}
\g > \sqrt{V_o}
\end{equation}
\end{proof}

\begin{prop} \label{bigg} Let $V(x)$ be nonsingular. Assume 
that ${}^{\gamma_o} \mathcal{H}$ admits a negative energy bound eigenstate, then so does ${}^{\gamma} \mathcal{H}$ 
for $\gamma > \g_o$.
\end{prop}
\begin{proof}
Let $\psi_o$ be a negative energy bound 
eigenstate of  ${}^{\gamma_o} \mathcal{H}$,  normalized such that $\psi_o(0)=1$, i.e., $\psi_o$ satisfies (\ref{ef})-(\ref{bcE}) 
for $\g_o$. Fix $\a \equiv \g - \g_o >0$ and let   $\phi_{\d}(x)$ be a smooth function of $(\d,x) \in (-\d_0,\d_o) \times (-\infty,\d_o)$ for 
some positive $\d_o$   
such that: $\phi_{\d}(-\d)=1$, $\phi_{\d}'(-\d)=0=\phi_{\d}''(-\d)$, $\phi_{\d}$ and $\phi'_{\d}$ are growing functions for  $x \in (-\d,0)$, 
$\phi'_{\d}(0)/\phi_{\d}(0)=\a$ and $\phi_{\d}(0)<2$. An example of such a function is 
\begin{equation}
\phi_{\d}(x)= \frac{\a (x + \d)^3}{\d^2(3-\a \d)} + 1, \;\;\; 0<\d<\tfrac{3}{2\a}=\d_o. 
\end{equation}
Define
\begin{equation}
\psi_{\d}(x) = \begin{cases} \psi_o(x) &, x \leq -\d \\ \psi_o(x) \phi_{\d}(x)
&, -\d<x\leq 0 \end{cases}
\end{equation}
Note that this function belongs to the domain of  ${}^{\gamma} \mathcal{H}$. An example of $\psi_o$ and the corresponding 
$\psi_{\d}$ is depicted in Figure \ref{f1}. \\
Since $V$ is continuous and bounded,  $\psi_o$ and $\psi'_o$ are bounded near $x=0$, and we may assume that 
\begin{equation}
|\psi_o(x)|<A, \;\;\; |\psi_o'(x)|<B, \;\;\; |V(x)|<C \;\;\;\text{ for }  x \in (-\d_o,0).
\end{equation}
This implies that 
\begin{equation}
|\psi_{\d}(x)|<2A, \;\;\; |\psi_{\d}'(x)|<2(B+\a A)\;\;\;\text{ for }  x \in (-\d_o,0).
\end{equation}
From these two equations and $\phi_{\d}(0)= (1-\a \d/3)^{-1}$ follows that 
\begin{align} \nonumber
\left<\psi_{\delta}, {}^{\g}\mathcal{H} \psi_{\delta} \right> -\left<\psi_o,  {}^{\g_o}\mathcal{H} \psi_o \right> = &
\gamma_o  - \g \phi_{\d}(0)^2 + \int_{-\delta}^0 [(\psi_{\delta}')^2-(\psi_o')^2 + V (\psi_{\delta}^2-\psi_o^2)] \\
&< \g_o - \g   (1-\a \d/3)^{-2}  +  \delta \left[ 4(B+ \alpha A)^2+B^2+5CA^2\right] \label{expect}
\end{align}
Note that $A, B$ and $C$ are independent of $\d$, then  
for $\d$ small enough the right side of 
(\ref{expect})  is negative. This implies that  $\left<\psi_{\delta}, {}^{\g}\mathcal{H} \psi_{\delta} \right> < 
\left<\psi_o, {}^{\g_o}\mathcal{H} \psi_o \right> <0$. Since $\psi_{\d}$ is  in the domain of  ${}^{\gamma} \mathcal{H}$, 
the spectrum of  ${}^{\gamma} \mathcal{H}$  contains a negative energy. 
\end{proof}

 \begin{figure}
    \includegraphics[width=7cm]{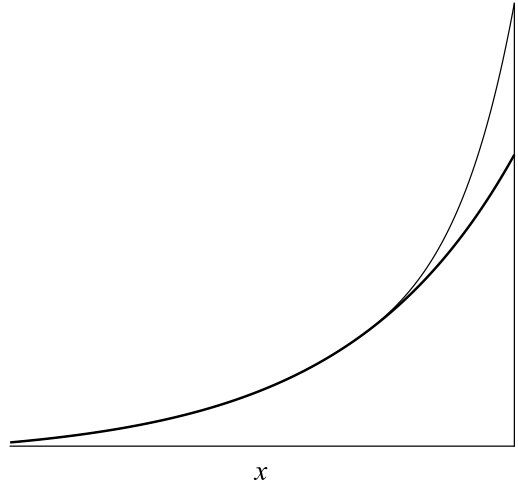}
    \caption{En example  of $\psi_o$ (thick line) and the corresponding $\psi_{\d}$ (thin line) used in Proposition \ref{bigg}}\label{f1}.
  \end{figure}

\subsection{Non negative nonsingular potentials} \label{nn}

For non negative nonsingular potentials a number of useful properties can be easily  proved:
\begin{prop}\label{stable}
Let  $V(x) \geq 0$ be a  nonsingular potential. 
 The self-adjoint extensions ${}^D\mathcal{H}$, ${}^N\mathcal{H}$ and 
${}^{\g}\mathcal{H}$  with $\gamma <0$ are positive 
definite.
\end{prop}
\begin{proof}
This follows from equations  (\ref{parts})  and  (\ref{negstab}).
\end{proof}
The following proposition  shows that if $V(x) \geq 0$ is nonsingular and $\g$ positive, ${}^{\gamma}\mathcal{H}$
 admits at most one negative energy eigenstate. It also  establishes  some properties 
of the corresponding  
eigenfunction. 

\begin{prop}\label{eigenf}
 Let  $V(x) \geq 0$ be a nonsingular potential.  Assume there is $E<0$ and   $\psi_E$ satisfying (\ref{ef})-(\ref{bcE}), 
then it follows that:
\begin{itemize}
\item[i)] $\psi_E$ has no roots, then we can choose it to be positive. 
\item[ii)]   $\psi_E$ grows monotonically from $0$ to $1=\psi_E(0)$ and $\psi_E'$ grows monotonically from $0$ to 
$\gamma =\psi'_E(0)$ in the 
interval $x \in (-\infty,0]$
\item[iii)] There is at most one $E<0$ and one $\psi_E$ for which   conditions (\ref{ef})-(\ref{bcE}) hold.\\
\end{itemize}
\end{prop}
\begin{proof} \mbox{}
\begin{itemize}
\item[i)] First note that the roots of $\psi_E$ are isolated points: if $x_o = \lim_{n \to \infty} x_n$ were roots of $\psi_E$, then 
$\psi_E(x_o)=0=\p_x \psi_E(x_o)$ and, since $\psi_E$ satisfies the second order equation (\ref{ef}), $\psi_E(x)=0$ for all $x$, 
which is a contradiction. Note also that there cannot exist 
 a sequence $x_n$ of consecutive roots such that $\lim_{n \to \infty} x_n = -\infty$, otherwise, 
 there would be 
a sequence $x'_n$ of   positive local maxima of $\psi_E$, $\lim_{n \to \infty} x'_n =-\infty$,  and/or a sequence $x''_n$ of  negative local minima 
of $\psi_E$, $\lim_{n \to \infty} x''_n =-\infty$. 
Both cases lead to a contradiction: 
 take, e.g., a sequence $x'_{n}$ of positive maxima: the conditions  
$\psi_E(x'_{n})>0$, 
$\p_x^2 \psi_E(x'_{n}) <0$, $\lim_{n \to \infty} x'_{n}= - \infty$ give 
\begin{equation}
0 \geq \lim_{n \to \infty} \frac{\p_x^2 \psi_E(x'_{n})}{\psi_E(x'_{n})}
= \lim_{j \to \infty} (V(x'_{n})-E) = -E,
\end{equation}
which contradicts $E<0$ (if we assume a sequence $x''_n$  of negative minima we get the same contradiction as, again, $\p_x^2 \psi_E(x''_{n})/\psi_E(x''_{n}) <0$). 
We may therefore assume, without loss of generality (i.e., replacing $\psi_E$ with $-\psi_E$ if necessary), that $\psi_E>0$ for large negative $x$. 
Now assume $\psi_E$ has roots, then there is a 
root $x_o$  with largest absolute value 
(the least upperbound   of the non empty bounded set $\{ z \leq 0 | \psi(x)>0 \text{ for } x \in (-\infty,z) \}$).  
 Since $\psi_E$ is square integrable, 
$\lim_{x \to -\infty} \psi_E(x)=0=\psi_E(x_o)$, and $\psi_E(x)>0$ in the interval $x \in (-\infty,x_o)$. 
It follows that there is 
a local maximum at $x_1 <x_o$, and this  leads us  back to the same type of 
contradiction as above: $0 > \p_x^2 \psi_E(x_1) / \psi_E(x_1) = V(x_1) - E$, however  $ V(x_1) - E>0$.
\item[ii)] As $\psi_E$ is in the domain of  ${}^{\g}\mathcal{H}$, $\int_{-\infty}^0 (\psi_E')^2 dx < \infty$, then 
$\lim_{x \to -\infty}\psi_E'(x)=0$. 
From the hypotheses and i) follows that $\psi_E >0$ for $x \in (-\infty,0]$, then $\p_x^2 \psi_E = (V-E) \psi_E >0$. This implies 
that $\p_x \psi_E$ grows monotonically from zero to $\p_x \psi_E(0)$. As $\p_x \psi_E>0$ everywhere, $\psi_E$ increases monotonically from 
zero to $\psi_E(0)$.
\item[iii)] Assume $0>E_2>E_1$ are two eigenvalues of ${}^{\gamma}\mathcal{H}$. 
 In view of ii) we may assume that $\psi_{E_i}(x)>0$ for all $x$, $i=1,2$. Equation (\ref{ef}) implies that the Wronskian
\begin{equation}\label{wronskian}
W = (\p_x \psi_{E_1}) \psi_{E_2} - (\p_x \psi_{E_2}) \psi_{E_1}
\end{equation}
satisfies $\p_x W = (E_2-E_1) \psi_{E_2} \psi_{E_1} >0$ for all $x$. This implies that $W$ grows monotonically. However 
$\lim_{x \to -\infty}W(x) = 0 = W(0)$. The uniqueness of $\psi_E$  follows from  (\ref{bcE}).
\end{itemize}
\end{proof}

\begin{prop} \label{grange}  Let  $V(x) \geq 0$ be a nonsingular potential and consider the operator ${}^{\gamma} \mathcal{H}$ for $\g>0$.
\begin{itemize}
\item[i)] If $V$ is non trivial, for small enough  $\g$ the spectrum of ${}^{\gamma} \mathcal{H}$ contains no negative eigenvalue. 
\item[ii)] There is a critical value $\g_c>0$ such that the set $\{ \g \;  | \; \text{the spectrum of }{}^{\gamma} \mathcal{H} 
\text{ contains a negative eigenvalue} \}$ is of the form $(\g_c, \infty)$.
\item[iii)] For $\g > \g_c$, the negative energy eigenvalue $E_{\g}$ satisfies $|E_{\g}| \leq 2 \g^2$.
\item[iv)] If $\int_{-\infty}^0 V dx < \infty$ then 
$
\g_c \leq  2 \int_{-\infty}^0 V dx.
$
\end{itemize}
\end{prop}
\begin{proof}\mbox{}
\begin{itemize}
\item[i)] Assume ${}^{\gamma} \mathcal{H}$ admits a negative energy $E_{\g}$ and let 
$\psi_{E_{\g}}$ be the eigenstate satisfying $\psi_{E_{\g}}(0)=1$. From Lemma \ref{eigenf} 
we know that 
$\psi_{E_{\g}}$ is positive, convex and monotonically growing in $(-\infty,0]$. In particular, $\psi_{E_{\g}}(x) > \g x+1$ 
for $-1/\g < x <0$ (Figure \ref{f2}), then 
\begin{equation} \label{con}
\g = \int_{-\infty}^0  \psi_{E_{\g}}'' (x)\; dx = \int_{-\infty}^0  (V(x)-E_{\g}) \psi_{E_{\g}}(x) \; dx > 
 \int_{-1/\g}^0  (V(x)-E_{\g}) (\g x +1) \; dx = h(\g) + k(\g),
\end{equation}
where, for $\g>0$ we defined 
\begin{equation} \label{hk}
h(\g) =  \int_{-1/\g}^0  V(x) (\g x +1) \; dx \geq 0  \;\;\;\text{ and } \;\;  k(\g) = - \frac{E_{\g} }{2 \g} >0.
\end{equation}
Note that $dh/d\g = \int_{-1/\g}^0  V(x)  x  \; dx \leq 0$, then $h$ is a positive, decreasing function of $\g$, 
diverging as $\g \to 0^+$ unless $\int_{-\infty}^0 V(x) \; dx$ is finite. It follows from (\ref{con}) and (\ref{hk})  that the 
existence of a negative energy eigenvalue implies 
$\g > h(\g)$, and this implies that 
$\g>\g_*$, where $\g_*>0$ is the only solution of $\g = h(\g)$. We conclude that, for small positive $\g$,  
${}^{\gamma} \mathcal{H}$ is positive definite.
\item[ii)] This follows from i) and Propositions \ref{robinst} and \ref{bigg}. For $\g=\g_c$ we expect the lowest energy eigenvalue to be $E=0$. 
\item[iii)] From (\ref{con}) and (\ref{hk}) follows that,  for the bound state, $\g \geq k(\g)=-E_{\g}/(2 \g)$. 
\item[iv)] For the trial function used in Proposition \ref{robinst} 
\begin{equation}
\left< \psi , {}^{\g}\mathcal{H} \psi \right> 
= -\frac{\g}{2} + \int_{-\infty}^0  V e^{2\g x} \;dx  < -\frac{\g}{2} + \int_{-\infty}^0  V  \;dx, 
\end{equation}
and this  is negative for 
\begin{equation} \label{notsharp}
\g > 2 \int_{-\infty}^0  V  \;dx.
\end{equation}
\end{itemize}
\end{proof}

 \begin{figure}
    \includegraphics[width=7cm]{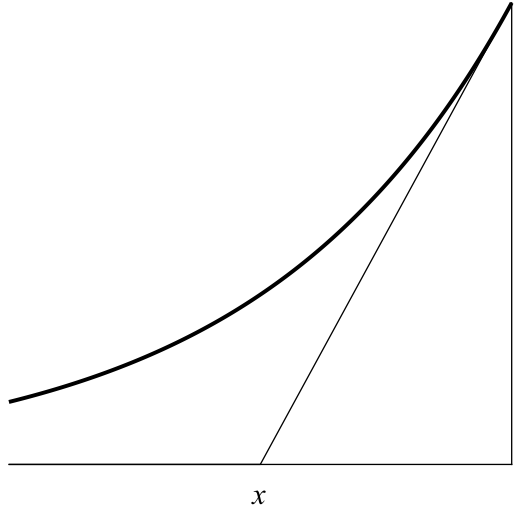}
    \caption{Graphs of $\psi_{E_{\g}}$ (thick line) and $\g x +1$ for  $x>-1/\g, $ zero otherwise 
 (thin line) used in Proposition \ref{grange}.i).} \label{f2}
  \end{figure}

For potentials  with finite integral, condition (\ref{notsharp}) assures the existence of negative energy states for ${}^{\g}\mathcal{H}$.  
The sharper bound 
\begin{equation} \label{sharp}
\g >  \int_{-\infty}^0  V  \;dx.
\end{equation}
was proved in \cite{Anabalon:2015vda} using different methods. We can prove (\ref{sharp}) by making the 
 assumption  that the spectrum of ${}^{\g}\mathcal{H}$ is continuous in $\g$.\\

\begin{prop}  Let  $V(x) \geq 0$ be a nonsingular potential and assume that $\lim_{\g \to {\g_c}^+} E_{\g} =0$, then
\begin{equation} \label{sharp2}
\g_c \leq \int_{-\infty}^0  V(x)\; dx,
\end{equation}
\end{prop}
\begin{proof}
For $\g>\g_c$ 
and the normalization $\psi_{E_{\g}}(0) =1$ 
\begin{equation} \label{eq}
\g  = \int_{-\infty}^0  (V(x)-E_{\g}) \psi_{E_{\g}}(x) \; dx. 
\end{equation}
Taking  the limit $\g \to {\g_c}^+$ and using Proposition \ref{eigenf}.ii) this gives
\begin{equation} \label{gcub}
\g_c = \lim_{\g \to {\g_c}^+} \int_{-\infty}^0  (V(x)-E_{\g}) \psi_{E_{\g}}(x) \; dx = 
\lim_{\g \to {\g_c}^+} \int_{-\infty}^0  V(x) \psi_{E_{\g}}(x) \; dx \leq 
\int_{-\infty}^0  V(x)\; dx,
\end{equation}
from where the existence of a negative energy in the spectrum for $\g$ satisfying (\ref{sharp}) follows. 
\end{proof}

\subsection{Energy considerations} \label{ec}

As pointed out in \cite{Ishibashi:2004wx}, no matter which self adjoint extension we choose for $\mathcal{H}$, there is always 
a notion of conserved energy for solutions of the wave equation (\ref{dd22})-(\ref{hlm2}) in the domain $x<0$ of 1+1 Minkowski 
spacetime for the solution (\ref{sol11})-(\ref{h}).
 If ${}^z \mathcal{H}$ is the chosen self adjoint extension ($z=D, N, \gamma$), the energy is defined 
as
\begin{equation} \label{energy1}
\mathcal{E}_z =\tfrac{1}{2} \left( \left< \p_t \phi, \p_t \phi \right> + \left< \phi, {}^z \mathcal{H} \phi \right> \right).
\end{equation}
Conservation of $\mathcal{E}_z$ follows from $[\p_t,{}^z \mathcal{H}]=0$ and the self adjointness of ${}^z \mathcal{H}$: 
\begin{equation}
\dot{\mathcal{E}_z} =  \left< \p_t \phi, \p_t^2 \phi \right> + \tfrac{1}{2} \left(\left<\p_t \phi, {}^z \mathcal{H} \phi \right>  
+ \left<  \phi, \p_t {}^z \mathcal{H} \phi \right>\right) =   \left< \p_t \phi, \p_t^2 \phi +{}^z \mathcal{H} \phi \right>  =0.
\end{equation}
Integrating (\ref{energy1}) by  parts we get (c.f. equation (\ref{parts}))
\begin{equation} \label{ez}
\mathcal{E}_z = - \tfrac{1}{2} z \phi^2  \big|_{x=0} + \tfrac{1}{2} 
\int_{-\infty}^0  [(\p_t \phi)^2+ (\p_x \phi)^2 + V \phi^2] \; dx \equiv \mathcal{E}_o  - \tfrac{1}{2} z \phi^2  \big|_{x=0},
\end{equation}
where $z = 0$ for Dirichlet or Neumann boundary conditions and $z=\gamma =(\p_x \phi /\phi)|_{x=0}$ for Robin boundary conditions.\\

The conservation of $\mathcal{E}_z$ in  (\ref{ez})  can also be derived 
by  applying Gauss' theorem to the conserved current $J_a = T_{ab} (\p/\p t)^b$,  where 
$T_{ab}$ is the energy  momentum tensor of 
the 1+1 field theory  (\ref{dd22})-(\ref{hlm2}) in the two dimensional (half) Minkowski space.
The region of integration is the one bounded 
by two $t=$constant surfaces. The integral at the horizon ($x=-\infty$) 
vanishes for fields $\phi$ which are square integrable on $t-$slices, so 
if we use inertial coordinates $(t,x)$ we get
\begin{equation} \label{flux}
0 = \int_{-\infty}^0  T_{tt}(t,x) \; dx - \int_{-\infty}^0  T_{tt}(t_o,x) \; dx- \int_{t_o}^t T_{tx}(t',0) \; dt', 
\end{equation}
where $T_{ab} = \p_a \phi \p_b \phi - \tfrac{1}{2} \eta_{ab} (\p_c \phi \p^c \phi + V \phi^2)$. Note that 
\begin{equation}
T_{tt}= \tfrac{1}{2}(\dot \phi^2 +{\phi'}^2 + V \phi^2),
\end{equation}
and $T_{tx} = \dot \phi \phi'$, then  
\begin{equation}  \label{robinflux}
\lim_{x \to 0^-} T_{tx}=\lim_{x \to 0^-} 
\dot \phi \phi'  \to \begin{cases} 0 &, \text{N or D boundary conditions} \\ \tfrac{1}{2} \gamma \p_t (\phi^2) &, \text{Robin boundary 
conditions.}
\end{cases}
\end{equation}
Thus, for N or D boundary condition the third
 term in (\ref{flux}) vanishes, there is no flux of energy at the conformal boundary,  and the canonical energy $\mathcal{E}_o$ defined in (\ref{ez}) is conserved.  
As seen from equation (\ref{robinflux}), the  form  (\ref{robin}) of 
the Robin boundary condition is crucial for the existence of the conserved quantity $z=\g$ in (\ref{ez}), as it allows  the flux 
of energy at infinity 
(the timelike boundary)  
to be  integrated,  reducing (\ref{flux}) to 
\begin{equation} \label{ez2}
\mathcal{E}_o(t)-\mathcal{E}_o(t_o) = \tfrac{\gamma}{2} [\phi^2(t,0)- \phi^2(t_o,0)],
\end{equation}
from where the conservation of $\mathcal{E}_z$ for the $z=\g$ case of (\ref{ez}) follows. 
Note from (\ref{robinflux}) and (\ref{ez2}) that  the change $\g \to -\g$ reverses the sign of the energy flow at $x=0$. 
Note also that a positive $\gamma$ in (\ref{ez2}) opens the possibility of an unbounded growth 
of the canonical energy $\mathcal{E}_o$ due to energy pumped in from the boundary (right hand side in (\ref{ez2})),
 but the fact that the 
critical value $\g_c$ for instability  is in general strictly positive implies that a flux of energy 
from the boundary is not a sufficient condition for instability. What happens for 
 $\g> \g_c$ is that the equation of motion allows that  $\phi(t,0)$ grew unbounded together 
with the energy $\sim \phi^2(t,0)$ pumped into the system, this being the mechanism driving  the instability. 
This can be understood with the help of the toy models of the following Section: equation (\ref{2dwe}) can be 
regarded as transverse oscillations of a semi infinite string  $x<0$ (with elastic $x-$dependent restoring forces when $V(x)>0$), 
and the boundary condition $\phi ' = \g \phi$ at $x=0$ is equivalent to an elastic restoring force at $x=0$ 
with elastic constant proportional to $-\g$. For positive $\g$, the force is not restoring but repulsive and proportional 
to $\phi(t,0)$ and, if $\g > \g_c$, it allows for solutions with $\phi(t,0)$ increasing without bound. This negative 
elastic potential energy  
 is compensated with a diverging positive energy in the string modes (see equation (\ref{ez2})).

\subsection{Toy models}\label{toy}

\subsubsection{Vanishing potential: a  semi infinite string}

Consider a string extending along the $x<0$ half axis  and  oscillating in the $(x,y)$ plane, and let $\phi(x,t)$ be the $y-$displacement 
at time $t$. For small $\p_x \phi$ we use the standard approximation $\p_x \phi = \tan(\alpha) \simeq \sin(\alpha)$ 
for the angle $\a$ of the string with respect to the horizontal at $x$ and calculate the vertical component of the tension $T$ as 
$T_y \simeq  T \p_x \phi$, then Newton's law applied to the piece of string extending from $x$ to $x + \Delta x$ gives  
\begin{equation}  \label{force}
T \left[\p_x \phi (x+\Delta x,t) - \p_x \phi (x,t) \right] = \rho \; \Delta x \; \p_t^2 \phi,
\end{equation}
where $\rho$ is the mass per unit length. Taking the limit $\Delta x \to 0$ we obtain the wave equation
\begin{equation} \label{wes}
\tfrac{1}{c^2}\,  \p_t^2 \phi - \p_x^2 \phi = 0, \;\;\; c^2 = T/\rho,
\end{equation}
which, after rescaling $t$,  has the form (\ref{dd22})-(\ref{hlm2}) with $\mathcal{H} = - \p_x^2$, that is $V=0$. \\

For a string with a right end at $x=0$, Newton's law applied to the $-\Delta x < x <0$ piece of the string 
\begin{equation}
f- T \p_x \phi(0-\Delta x,t)  = \rho\; \Delta x \; \p_t^2 \phi,
\end{equation}
gives, after taking  the limit $\Delta x \to 0$, the force   $f$  applied to the string end at $x=0$:
\begin{equation}
f = T \; \p_x \phi(0,t).
\end{equation}
Dirichlet boundary conditions $\phi(0,t)=0$ corresponds to fixing the string at the $x=0$ end,
 Neumann boundary conditions $\p_x \phi(0,t)=0$ to leaving 
it free ($f=0$), and Robin boundary conditions $\p_x \phi(0,t)=\g \phi(0,t)$ to subjecting the end of the string 
to an elastic  force 
\begin{equation}
f = \gamma T \, \phi(0,t)
\end{equation}
which corresponds to a spring with elastic constant $k= -\gamma T$ {\em if $\gamma<0$}, and gives a  
``repulsive elastic force'' if $\gamma>0$. \\

\noindent
{\em Dirichlet boundary conditions:}
The eigenfunctions of ${}^D\mathcal{H}$ are 
\begin{equation}
\psi_{(E=k^2)} = \sqrt{\tfrac{2}{\pi}} \; \sin(kx), \;\; k>0.
\end{equation}
These were normalized such that
\begin{equation} \label{ss}
\left< \psi_{(E=k^2)},\psi_{(E=l^2)} \right> = \frac{2}{\pi} \int_{-\infty}^0 \; \sin(kx) \; \sin(lx) \; dx = \delta(k-l), \;\;\;\; (k,l >0).
\end{equation}
Functions $\psi$ satisfying the  Dirichlet boundary conditions extend naturally  to odd functions in $\mathbb{R}$. The completeness of the above 
set of eigenfunctions then follows from the sine Fourier transform theorem for  odd functions:\\
\begin{equation}
\psi(x) = \sqrt{\tfrac{2}{\pi}} \int_0^{\infty} A_k \sin(kx) \; dk,
\end{equation}
where
\begin{equation}
A_k = \sqrt{\tfrac{2}{\pi}}  \int_{-\infty}^0 \psi(x) \; \sin(kx) \; dx 
\end{equation}
The conserved energy in this case  is $\mathcal{E}_o$ defined in (\ref{ez}). 
Energy can be transferred among the different normal modes of the string but remains conserved. This is due to the fact that 
the external force that keeps 
the $x=0$ end of the string fixed does no work on the string.\\

\noindent
{\em Neumann boundary conditions:}
The eigenfunctions of ${}^N\mathcal{H}$ are 
\begin{equation}
\psi_{(E=k^2)} = \sqrt{\tfrac{2}{\pi}} \; \cos(kx), \;\; k>0.
\end{equation}
These were normalized such that
\begin{equation}\label{cc}
\left< \psi_{(E=k^2)},\psi_{(E=l^2)} \right> = \frac{2}{\pi} \int_{-\infty}^0 \; \cos(kx) \; \cos(lx) \; dx = \delta(k-l), \;\;\;\; (k,l >0).
\end{equation}
Functions $\psi$ satisfying the Neumann boundary conditions are naturally extended to even functions in $\mathbb{R}$. The completeness of the above 
set of eigenfunctions then follows from the cosine Fourier transform  theorem of even functions.\\
The conserved energy in this case  is again $\mathcal{E}_o$ defined in (\ref{ez}).
Vibrational energy can be transferred among the different normal modes of the string but remains conserved. This is due to the fact that no external  force is acting on 
the $x=0$ end of the string.\\

\noindent
{\em Robin boundary conditions:}
For any value of $\gamma$, the generalized eigenfunctions corresponding to the continuum spectrum are 
\begin{equation} \label{cs}
\psi_{(E=k^2)} = \sqrt{\frac{2}{\pi}} \left( 1+\frac{k^2}{\g^2} \right)^{-1/2} \left[\sin (kx) + \frac{k}{\g} \cos(kx)\right], \; k>0.
\end{equation}
These were normalized such that 
\begin{equation} \label{norm}
\int_{-\infty}^0 \psi_{(E=k^2)}\,  \psi_{(E=l^2)} \, dx = \delta(k-l).
\end{equation}
To verify (\ref{norm}) we use the integrals (\ref{ss}) and (\ref{cc}) together with the distributional identity 
\begin{equation} \label{sc}
\int_{-\infty}^0 \; \sin(lx) \, \cos(kx) \; dx =\frac{l}{(k+l)(k-l)}.
\end{equation}
From Proposition \ref{grange}.iv we expect instabilities when $\gamma >0$. This is trivially verified, 
 the only (Proposition \ref{eigenf}.iii)  bound state of  ${}^{\g}\mathcal{H}$ for $ \g>0$ is 
\begin{equation} \label{uns}
\psi_{(E=-\g^2)} = \sqrt{2 \g} \; e^{\g x},
\end{equation}
where we have normalized such that the integral of $\psi_{(E=-\g^2)}^2$ equals one. 
This function has to be added to the set (\ref{cs}) to form a complete orthonormal 
set for the domain of ${}^{\g}\mathcal{H}$ when $\g>0$. 
To verify that (\ref{uns}) is orthogonal to the functions (\ref{cs}) we use 
\begin{equation}
\int_{-\infty}^0 e^{\g x} e^{-ikx} dx = \frac{1}{\g-ik}, \;\;\; \g>0.
\end{equation}

We can also understand the expansion of functions satisfying (\ref{robin}) in terms of ordinary Fourier transforms 
by means of the following observation: the function $\tau = \psi- \p_x\psi /\g$ vanishes at $x=0$, and can 
naturally be extended to an odd function in $\mathbb{R}$, for which the sine Fourier representation is possible:
\begin{equation}\label{rep}
\tau(x) = \sqrt{\tfrac{2}{\pi}} \int_0^{\infty} \tau_k \sin(kx) \; dk,\;\;\;
\tau_k = \sqrt{\tfrac{2}{\pi}}  \int_{-\infty}^0 \tau(x) \; \sin(kx) \; dx. 
\end{equation}
$\psi$ can be easily recovered from $\tau$:
\begin{equation}
\psi(x) = e^{\g x} \left[ \psi(0) + \g \int_{x}^0 \tau(u) e^{-\g u} du \right]
\end{equation}
If we use the representation (\ref{rep}) of $\tau$ in the above equation we arrive at the following 
expansion of $\psi$ in terms of $\psi_{(E=-\g^2)} $ and the $\psi_{(E=k^2)}$ above: 
\begin{equation}
\psi(x) = \left[ \psi(0) - \sqrt{\tfrac{2}{\pi}} \int_0^{\infty} dk  \frac{\g \, k \, \tau_k}{\g^2 + k^2} \right] e^{\g x}
+ \sqrt{\tfrac{2}{\pi}}  \int_0^{\infty}  \frac{\g^2 \, \tau_k}{\g^2 + k^2} \; \left(\sin (kx) + \frac{k}{\g} \cos(kx)\right) dk.
\end{equation}
It is a non trivial fact that the coefficient between square brackets above vanishes if $\g<0$.\\

The conserved energy $\mathcal{E}_{\g}$ (equation (\ref{ez})) contains two terms: the string vibrational energy $\mathcal{E}_o$
 and the potential energy $-\g\phi^2(x=0,t)/2$ 
of the  force acting on the $x=0$ end of the string. For negative $\g$, this is an ordinary elastic force pulling towards $\phi=0$: we may imagine that the $x=0$ end 
of the string is attached to a spring of elastic constant $-\g$, then energy flows from the string to the spring and viceversa in such a way that the 
total energy $\mathcal{E}_{\g}$  remains constant. Since the spring
 potential energy is positive definite, the amount of energy the spring  can transfer to the string is finite; this 
keeps the string vibrations bounded. For positive $\g$, instead, the force at the $x=0$ end is ``repulsive elastic'', pushing away the string end with 
an intensity that {\em increases} as $\phi(x=0)$ grows. The repulsive elastic potential $-\g \phi^2/2$ is unbounded from below, it can feed 
the string with an unlimited amount of energy and produce unbounded oscillations.\\

The string analogy can be extended to the $V(x) \geq 0$ case by assuming that, besides the string tension, there is an $x-$dependent  restoring elastic force 
pulling the string to the $\phi=0$ configuration. In this  case,  instead of (\ref{force}) we have
\begin{equation}  \label{force2}
T \left[\p_x \phi (x+\Delta x,t) - \p_x \phi (x,t) \right] -   V(x)  \; T \; \Delta x \; \phi(x,t) = \rho \; \Delta x \; \p_t^2 \phi
\end{equation}
which, taking the limit $\Delta x \to 0$ in (\ref{force2}) and rescaling $t$ gives  (\ref{dd22})-(\ref{hlm2}).

\subsubsection{A step potential} \label{stepot}

Consider now the case 
\begin{equation}
V(x) = \begin{cases} 0 &x<-a \\ V_o  &a<x \leq 0 \end{cases}
\end{equation}
where $V_o$ and $a$ are positive, and assume ${}^{\gamma}\mathcal{H}$ has  a negative  energy eigenvalue $E=-\alpha^2, \a>0$. 
(Note that a single finite discontinuity in $V$ does not invalidate the results of the previous Section.) The wave function 
is proportional to 
\begin{equation} \label{eigstep}
\psi(x) = \begin{cases} \exp(\a (x+a)) &, x<-a \\  \cosh(\b(x+a))+ \frac{\a}{\b} \;  \sinh(\b (x+a)) &, -a<x \leq 0 \end{cases}
\end{equation}
where 
\begin{equation} \label{beta}
\beta = \sqrt{\alpha^2 + V_o}.
\end{equation}
This function  is $C^1$ and satisfies the properties in Proposition \ref{eigenf}. Imposing (\ref{robin}) on (\ref{eigstep}) 
gives the following relation
\begin{equation} \label{lam}
\gamma = \beta \; \frac{\beta \tanh(\beta a) + \alpha}{\beta+ \alpha \tanh(\beta a)} 
\end{equation}
Inserting (\ref{beta}) in (\ref{lam}) we find that $d\gamma/ d\alpha >0$ for positive $\alpha$. Since $\gamma \sim \alpha$ for $\alpha \to \infty$, we conclude 
that, as $\alpha$ goes from zero to infinity, $\gamma$ ranges from 
\begin{equation} \label{lc}
\gamma_c = \sqrt{V_0} \tanh(\sqrt{V_0} \;a)
\end{equation}
 to infinity. This example is useful because $\g_c$ introduced in Proposition \ref{grange}.ii can be explicitly calculated, equation (\ref{lc}). 
If $\gamma > \gamma_c$, (\ref{lam}) has a unique solution $\alpha$, which gives a unique bound state, of energy $E=-\alpha^2$. 
On the other hand, if $\gamma \leq \gamma_c$, there are no bound states.
The bounds (\ref{c1}) and (\ref{sharp}) can be easily checked in this example; moreover, the example  shows  that   
(\ref{sharp}) 
cannot be improved as, 
for small $a$ (\ref{lc}) gives
\begin{equation}
\g_c = a V_o + \mathcal{O}(a^3) = \int_{-\infty}^0 V \; dx + \mathcal{O}(a^3). 
\end{equation}

\section{Robin instabilities in SAdS${}_4$} \label{robin4}

In Section \ref{sf4d} we proved that Dirichlet is the only possible 
choice  of boundary condition at infinity for a massless scalar field and that, moreover,
the field is stable, in the sense that exponentially growing modes are not allowed. 
The uniqueness of dynamics is due to the fact that the mode potentials (\ref{VKG}) are LP at $x=0$. 
For Maxwell fields and gravitational perturbations, the mode potentials are (\ref{max4}) and (\ref{2drwv}) and (\ref{zp}) 
respectively, and they are 
all  LC at $x=0$. Besides, these potentials are continuous for $ x \in (-\infty,0]$, 
satisfy (\ref{vrh}) and behave near $x=0$ as in (\ref{vc}). Thus, the results in Section \ref{robinst-section} apply 
to these cases. In particular,  instabilities 
are to be expected for certain Robin boundary conditions.  In this Section we show that Robin boundary conditions 
arise naturally for Maxwell fields and gravitational perturbations on SAdS${}_4$ and explore the associated 
 ``Robin instabilities''. \\

\subsection{Maxwell fields}

In what follows  we show how the mode conserved energy resulting from equations  (\ref{ez}) and (\ref{ez2}) is connected to  
the electromagnetic field energy. The electromagnetic 
field modes are labeled by $(p=\pm,\ell,m)$ according to (\ref{maxo}) and (\ref{maxe}), 
 and (\ref{ez}) and (\ref{ez2}) imply that, {\em for any function} $Q(p,\ell,m)$ 
\begin{multline} \label{gauss2.5}
\sum_{(\ell,m,p=\pm)} Q(p,\ell,m) \int_{-\infty}^0  \left[\dot \phi_{(p,\ell,m)}^2+  {\phi'}_{(p,\ell,m)}^2 + V_{\ell}^{Max} \phi_{(p,\ell,m)}^2 \right] \; dx \; 
\bigg|_t =
\\ \sum_{(\ell,m,p=\pm)}   Q(p,\ell,m) \int_{-\infty}^0 
\left[\dot \phi_{(p,\ell,m)}^2+  {\phi'}_{(p,\ell,m)}^2 + V_{\ell}^{Max} \phi_{(p,\ell,m)}^2 \right] \; dx \; \bigg|_{t_o}\\
+ \sum_{(\ell,m,p=\pm)}  Q(p,\ell,m) \gamma_{(p,\ell,m)} (\phi^{}_{(p,\ell,m)}(s,0))^2
\bigg| ^{s=t}_{s=t_0}.
\end{multline}
We will show that the choice $Q(p,\ell,m)=\ell(\ell+1)/2$ in (\ref{gauss2.5}) gives the balance 
equation for the change of electromagnetic energy due to the energy flow from infinity. Energy 
is measured using the conservation of $J_{\a}:=T_{\a\b}\xi^{\b}$ for  $\xi^{\b} \p / \p x^{\b}=\p/\p t$ the timelike 
Killing vector field and $T_{\a\b}$ the 
 energy momentum tensor of the Maxwell field, 
\begin{equation}\label{emtensor}
 T_{\a\b}=\tfrac{1}{4\pi}(F_{\a\g}F_{\b}{}^{\g}-\tfrac{1}{4}g_{\a\b}F_{\g\d}F^{\g\d}).
\end{equation}
A similar result could be obtained in the gravity case, with $T_{\a\b}$  the effective energy-momentum 
tensor quadratic in the first order fields that sources the second order perturbation equations,  
however,  we have found that this calculation becomes 
unwieldy even using symbolic manipulation computing. \\

Applying Gauss' theorem to $J^{\a}$ in a region $\Omega$ of the spacetime limited by two $t=$ constant surfaces gives  
\begin{equation} \label{gauss1}
 0=\int_{\Omega}\c_{\a}J^{\a}=\int_{\partial\Omega}J_{\a}n^{\a}=
 \int_{\Sigma_t}J_{\a}n^{\a}_{\Sigma_t}+\int_{\Sigma_{t_0}}J_{\a}n^{\a}_{\Sigma_{t_o}}+\int_{\mathcal{I}_t}J_{\a}n^{\a}_{\mathcal{I}}
\end{equation}
where ${\mathcal{I}_t}$ is the $R \to \infty$ limit of an $r=R$ hypersurface extending from $\Sigma_{t_o}$ to $\Sigma_{t}$, $t>t_o$. 
The induced volume elements on the hypersurfaces are understood on the integrals, the outer pointing unit normal
vectors are 
$n^{\a}_{\Sigma_t} \p/\p x^{\a}= f^{-1/2} \p/ \p t$,  $n^{\a}_{\Sigma_t} \p/\p x^{\a}
= -f^{-1/2} \p/ \p t$ and $n^{\a}_{\mathcal{I}_t}\p/\p x^{\a} =  f^{1/2} \p / \p r$. 
We can therefore rewrite (\ref{gauss1}) as 
\begin{multline} \label{gauss2}
 \int_{-\infty}^0  dx \int_{S^2} T_{tt}(t,r,\theta,\phi) \sin \theta\,  d \theta\,  d \phi
 = \int_{-\infty}^0    dx \int_{S^2} T_{tt}(t_o,r,\theta,\phi) \sin \theta\,  d \theta\,  d \phi \\ +
\lim_{R \to \infty} R^2 f(R) \int_{t_o}^t dt \int_{S^2} T_{tr}(t, R,\theta,\phi) \; \sin(\theta)\, d\theta \, d\phi,
\end{multline}
where the second term on the right hand side  gives 
the failure for the ``standard energy'' to be conserved. Inserting $A_{\b} = A_{\b}^{(-)}+A_{\b}^{(+)}$ given in 
(\ref{maxo})-(\ref{maxe}) into $F=dA$ and (\ref{emtensor}) gives (a prime denotes $\p_x= f \p_r$)
\begin{equation}
 (J^{(\ell,m)}_{\a}n^{\a}_{\mathcal{I}})\big|_{r=R}= \sqrt{f} \; T_{tr}\big|_{r=R}=\frac{1}{4\pi R^2 \sqrt{f(R)}}
 \left[\dot\phi^{+}_{(\ell,m)}\phi'^{+}_{(\ell,m)}+\dot\phi^{-}_{(\ell,m)}\phi'^{-}_{(\ell,m)} \right]
 \left[(\p_{\theta}S_{(\ell,m)})^2+\frac{1}{\sin^2\theta}(\p_{\phi}S_{(\ell,m)})^2\right].
\end{equation}
Note that in the $R \to \infty$ limit, assuming Robin boundary conditions gives 
 limit $\dot\phi_{(p,\ell,m)}\phi'_{(p,\ell,m)} =  \tfrac{d}{dt} [\g_{(p,\ell,m)} \phi_{(p,\ell,m)}^2/2]$ ($\g_{(p,\ell,m)}$ the 
Robin 
constant (\ref{robin}) for the mode $(p=\pm,\ell,m)$). This fact allows to integrate the flux at $\mathcal{I}_t$. Note also that 
changing the sign of $\g_{(p,\ell,m)}$ reverses the direction of the flux at $\mathcal{I}_t$. \\

To proceed we use that 
\begin{multline}
\int_{S^2}  \left[(\p_{\theta}S_{(\ell,m)})^2+\frac{1}{\sin^2\theta}(\p_{\phi}S_{(\ell,m)})^2\right] \sin(\theta) \; d\theta\, d\phi\\= \int_{S^2}  
\hd^{A}  S_{(\ell,m)} \hd_{A} S_{(\ell,m)} 
= -  \int_{S^2}S_{(\ell,m)} \hd^A \hd_A S_{(\ell,m)} = 4 \pi \ell (\ell+1),
\end{multline}
where $\hd_{A}$ is the covariant derivative on $S^2$ and  we used the orthonormality of the $S_{(\ell,m)}$. After a lengthy calculation, 
we find that (\ref{gauss2}) reduces  to  (\ref{gauss2.5}) with $Q(p,\ell,m)=\ell(\ell+1)/2$, as anticipated. \\

We now comment briefly on  the stability of Maxwell fields. Note that $V^{Max}_{\ell}$ is nonsingular and nonnegative. 
Thus,  the results in Propositions 1-6 apply. In particular: i) the field is stable 
if either Dirichlet, Neumann or Robin boundary conditions with negative $\g_{(p=\pm,\ell,m)}$ 
are chosen {\em for every} $ \phi^{p=\pm}_{(\ell,m)}$ (Proposition 3) and 
ii) the field is unstable if there is a mode $(p,\ell,m)$ for which $\g_{(p,\ell,m)} > \sqrt{\ell(\ell+1)} / r_h$ 
(Proposition 1). For the characteristics of the unstable modes, the remaining Propositions apply. \\

Note the rather indirect relation between the boundary conditions  on the electromagnetic field $F_{\a \b}$ 
and those  on the mode master variables $\phi^{\pm}_{(\ell,m)}$: from
(\ref{maxo}), $F =dA$ with 
\begin{equation}
A =  \sum_{(\ell,m)} \left[ \left(\p_x \phi^{(+,\ell,m)} dt +  \p_t \phi^{(+,\ell,m)}dx \right) S_{(\ell,m)} 
+  \phi^{(-,\ell,m)} \left( \p_{\phi} S_{(\ell,m)}/\sin(\theta) \; d\theta - \sin(\theta)  \p_{\theta} 
S_{(\ell,m)} \; d\phi \right)
\right]
\end{equation}
As an example, the condition  $F(\cdot,\p / \p x)=0$ at the boundary implies Dirichlet boundary conditions for 
the $\phi^{+}_{(\ell,m)}$ and Neumann boundary conditions on the $\phi^{-}_{(\ell,m)}$.

\subsection{Linearized gravity}

Let $C_{\a \b \g \d},$ be the Weyl tensor, ${}^*C^{\a \b \g \d}$ its dual. Define the algebraic curvature scalars 
\begin{align} \nonumber 
Q_+& = \tfrac{1}{48} C^{\a \b \g \d} C_{\a \b \g \d}, \\
Q_- &= \tfrac{1}{48} {}^*C^{\a \b \g \d} C_{\a \b \g \d}, \label{acs} 
\end{align}
and the differential curvature scalar 
\begin{equation} \label{dcs}
X = \frac{1}{720} \left( \nabla_{\e} C_{\a \b \g \d} \right)  \left( \nabla^{\e} C^{\a \b \g \d} \right).
\end{equation}
For S(A)dS${}_4$ these  fields are 
\begin{equation}
Q_+^{SAdS}=\frac{M^2}{r^6}, \;\;  Q_-^{SAdS}=0, \;\; X^{SAdS}=\frac{M^2}{3r^9} (\Lambda r^3-3r+6M).
\end{equation}
Consider the first order perturbation of these fields, $\d Q_+$, $\d Q_-$ and $\d X$.
 From symmetry arguments one can show that for odd perturbations only $\d Q_- \neq 0$ 
whereas for even perturbations $\d Q_- = 0$ while $\d Q_+$ and $\d X$ are nonzero. 
The scalar fields 
\begin{equation} \label{Gm}
G_-= \d Q_-
\end{equation}
and 
\begin{equation} \label{Gp}
G_+ = (9M-4r+\Lambda r^3) \d Q_+ + 3 r^3 \d X,
\end{equation}
are gauge invariant and encode all the gauge invariant information of the perturbation, 
in particular, it is possible to reconstruct the metric perturbation in a chosen gauge from the $G_{\pm}$ fields  
 \cite{Dotti:2013uxa,Dotti:2016cqy}. \\

 Schwarzschild black hole stability studies prior to \cite{Dotti:2013uxa} 
where limited to placing pointwise bounds on the  Regge-Wheeler and Zerilli master fields $\phi^{\pm}_{(\ell,m)}(t,r)$
(Wald, reference \cite{wald-s}) or analyzing the large $t$ decay of these fields (Price, \cite{Price:1971fb}, 
Brady et al \cite{Brady:1996za}).
The  $\phi^{\pm}_{(\ell,m)}$ enter the metric perturbation (in, say, the Regge-Wheeler gauge), in a series of the form
\begin{equation}\label{sov}
h_{\a \b} = \sum_{(\ell,m, p=\pm)}  D_{\a \b}^{(\ell,m,p)} [\phi^{(p)}_{(\ell,m)}, S_{(\ell,m)}]
\end{equation}
where the differential operators $D_{\a \b}^{(\ell,m,p)}$ are second order and linear. The separation of variables 
in (\ref{sov}) makes the linearized Einstein equations equivalent to the 1+1 wave equations satisfied by the 
$\phi^{\pm}_{(\ell,m)}(t,r)$ and the spherical harmonic equation satisfied by the $ S_{(\ell,m)}(\theta,\phi)$. 
It is clear from (\ref{sov}) that the relation of the $\phi^{(p)}_{(\ell,m)}$ to measurable perturbation effects is remote: 
four derivatives of these fields enter a single harmonic component of the curvature. Therefore
the boundedness of isolated  $\phi^{\pm}_{(\ell,m)}(t,r)$'s 
fields tells us little about the magnitude of the perturbation. \\

The nonmodal stability concept introduced in \cite{Dotti:2013uxa} is based on the pointwise boundedness (and 
decay, see \cite{Dotti:2016cqy}) of the $G_{\pm}$, which   are measurable geometric quantities on the 4D background spacetime 
that properly record the effect of the perturbation on the geometry.  
Not only it is established that there is a large $t$ 
 decay of the perturbed black hole to a member of the Kerr -dS family, but also that 
there are no transient growths of the $G_{\pm}$, something that  modal stability cannot rule out. 
Examples of modally 
 stable systems for which the isolated modes  decay exponentially with $t$  and yet  measurable quantities 
experience
 large transient growths 
are seen, e.g., in wall bounded shear flows  (see, e.g. \cite{schmid}). \\

 By iterating the linearized Einstein equations we arrive, after some work,  to a relatively simple 
{\em on shell}  form of $G_{\pm}$. 
For $G_-$ we find 
\begin{equation}\label{G-2}
G_- = -\frac{6M}{r^7} \sqrt{\frac{4 \pi}{3}} \sum_{m=1}^3 j^{(m)} S_{(\ell=1,m)} - \frac{3M}{r^5} 
\sum_{\ell>1,m} \frac{(\ell+2)!}{(\ell-2)!} \frac{\phi^-_{(\ell,m)}}{r}S_{(\ell,m)} ,
\end{equation}
where the first term contains the $\ell=1$ static contribution (not considered in the previous sections), 
$j^{(m)}, m=1,2,3$ being  the components of the perturbed black hole angular momentum. 
In view of (\ref{G-2}), 
 $r^5  G_-$ satisfies the four dimensional Regge-Wheeler equation (\ref{4drwe}). For the static 
$\ell=1$ in (\ref{G-2}) term  this  can be checked by a direct 
calculation, for the $\ell>1$ series this follows from   the spherical harmonic equation and the form of the 
$V^{(-)}_{\ell}$ potentials that enter the 1+1 wave equation satisfied by $\phi^-_{(\ell,m)}$, which 
contains the  required  $\ell (\ell+1)/r^2$ term, see (\ref{2drwv}). \\
Due to the intricate  $\ell$ dependence of the $V^{(+)}_{\ell}$ in (\ref{zp}) no 
similar construction leading to a four dimensional wave equation 
can be made using the  Zerilli fields $\phi^+_{(\ell,m)}$. These fields enter the on shell expression of 
$G_+$ as follows:
\begin{equation} \label{G+}
G_+ = - \frac{2M\; \d M}{r^5} + \frac{M}{2r^4} \sum_{\ell \geq 2} \frac{(\ell+2)!}{(\ell-2)!} \left[ f \p_r + Z_{\ell} \right] \phi^+_{(\ell,m)} S_{(\ell,m)},
\end{equation}
where $\d M$ is the mass variation that comes from the $\ell=0$ even perturbation (which is time independent and was not 
considered in the previous sections),  and 
\begin{equation} 
Z_{\ell}  = \frac{2M \Lambda r^3+\mu r (r-3M)-6M^2}{r^2 (\mu r + 6M)}, \;\; \mu = (\ell-1)(\ell+2).
\end{equation}
To see how natural is imposing  Robin boundary conditions on the  $\phi^{\pm}_{(\ell,m)}$ 
note that  Dirichlet conditions on the $G_{\pm}$  are  equivalent to 
 mixed  Dirichlet/Robin conditions on the $\phi^{\pm}_{(\ell,m)}$:
\begin{equation}
\phi^{-}_{(\ell,m)}\Big|_{x=0}=0, \;\;\;\; 
\frac{\p_x \phi^{+}_{(\ell,m)}}{\phi^{+}_{(\ell,m)}}\Big|_{x=0}=-\frac{2M\Lambda}{(\ell-1)(\ell+2)},
\end{equation}
whereas 
imposing Neumann or Robin boundary conditions on the $G_{\pm}$ gives Robin conditions on the $\phi^{\pm}_{(\ell,m)}$. 
The suitability of  a boundary condition  depends on the problem at hand, and 
for linear stability studies  it should be kept in mind that 
the fields $\phi^{\pm}_{(\ell,m)}$, although convenient to disentangle the linearized Einstein equations, are not relevant since 
they are not 
directly measurable 
quantities. \\

Another context where Robin boundary conditions on the $\phi^{\pm}_{(\ell,m)}$ arise is that of the AdS/CFT correspondence, 
under which the unperturbed background corresponds to a perfect fluid in the boundary and one is interested 
in metric perturbations that, in a preferred gauge, vanish in large $r=$constant surfaces 
 \cite{Bakas} \cite{Hussain}.\\
 The 
metric for the CFT is conformally related to that induced on large $r$ surfaces by
\begin{equation}\label{cft}
ds^2_{\infty} = \lim_{r \to \infty} \left( - \frac{3}{\Lambda r^2} ds_r^2 \right).
\end{equation}
In the unperturbed background $ds^2_r$ is obtained by setting $dr=0$ in (\ref{metric})-(\ref{f4m}) and 
the above limit gives  
\begin{equation}\label{ucftm}
ds^2_{\infty} = -dt^2 - \frac{3}{\Lambda} (d\theta^2 + \sin^2 (\theta) d\phi^2).
\end{equation}
For even perturbations in the Regge-Wheeler gauge,
  (\ref{cft}) gives (see equations (3.46) in  \cite{Bakas}, (129) in \cite{Dotti:2016cqy}, where $\mathcal{J}$ is defined) 
\begin{equation}\label{pcftm}
ds^2_{\infty} = -dt^2 -\frac{3}{\Lambda}  (1+ \tfrac{1}{2} \lim_{r \to \infty}\mathcal{J}) (d\theta^2 + \sin^2 (\theta) d\phi^2).
\end{equation}
Thus,  for the induced metric not to be perturbed we require that the  $(\ell,m)$ harmonic component  
$ \mathcal{J}_{(\ell,m)}$ 
of $\mathcal{J}$ vanishes for large $r$. 
 From  equations (151), (153)  and (154) in \cite{Dotti:2016cqy} we find that 
\begin{equation}
 \mathcal{J}_{(\ell,m)} \propto 2 f(r) \; \p_r \zeta^+_{(\ell,m)} +  \frac{\ell (\ell+1)}{r}\zeta^+_{(\ell,m)},
\end{equation}
where
\begin{equation}
\zeta^+_{(\ell,m)} = \left( (\ell+2)(\ell-1) + \frac{6M}{r} \right) \phi^+_{(\ell,m)}. 
\end{equation}
Keeping the conformal boundary metric unperturbed  will then 
 impose  the following  Robin boundary condition for $\phi^+_{(\ell,m)}$ 
at $x=0$:
\begin{equation}
\p_x \phi^+_{(\ell,m)} = -\frac{2 M \Lambda}{ (\ell+2)(\ell-1)} \phi^+_{(\ell,m)}.
\end{equation}

\subsubsection{Explicit unstable modes}

The following field was reported in \cite{Dotti:2016cqy}  as an unstable solution for the even gravitational perturbations of SAdS${}_4$ 
(equations (\ref{dd22})-(\ref{hlm2}) with potential (\ref{zp})) 
satisfying Robin boundary conditions:
\begin{equation} \label{un+}
\phi^{+ \;unst}_{(\ell,m)} = \chi^+_{\ell}(r) \exp( w_{\ell} \, t)
\end{equation}
where 
\begin{equation} \label{chi+}
\chi^+_{\ell}(r)  =  \frac{r \exp( w_{\ell} \, x )}{(\ell+2)(\ell-1)r+6M}, 
\end{equation}
$x$ is the radial coordinate defined in  (\ref{x4}), and 
\begin{equation}\label{wc}
w_{\ell} = \frac{1}{12M} \frac{(\ell+2)!}{(\ell-2)!}.
\end{equation}
$\chi^+_{\ell}(r)$ defined in (\ref{chi+})  satisfies 
 $\mathcal{H}^+ \chi^+_{\ell} 
= -w_{\ell}^2 \chi^+_{\ell}$ where $\mathcal{H}^+_{\ell}$ is the Hamiltonian  for even/scalar 
gravitational perturbations.  
This equation is satisfied 
{\em for any value} of $M$ and $\Lambda$, as long as $x(r)$ in (\ref{chi+}) satisfies 
$dx/dr=1/f$ (c.f. equation (\ref{x4})) with the appropriate parameters. 
For  $\Lambda=0$, this solution was  found by Chandrasekhar \cite{ch1} when looking for {\em algebraically special} perturbations: those  with the property that 
the first order variation 
 of one of the Weyl scalars $\Psi_0$ or $\Psi_4$ vanishes. 
 A  linearly independent solution with the same negative energy  is \cite{ch1}
\begin{equation} \label{tau+}
{}_{r_o}\tau_{\ell}^+(r) = \chi^+_{\ell}(r) \int^r_{r_o} \frac{dr'}{f(r') (\chi^+_{\ell}(r'))^2}.
\end{equation}
Note that changing $r_o$ above adds a term proportional to  $\chi^+_{\ell}(r)$.
 \\

Two linearly independent solutions of  $\mathcal{H}_{\ell}^- \psi^-_{\ell} 
= -w_{\ell}^2 \psi^-_{\ell}$  for the odd (vector) Hamiltonian
 with the same negative energy $E=-{w_{\ell}}^2$  are \cite{ch1}
\begin{equation}\label{uns-}
\chi^-_{\ell}(r) = \frac{1}{\chi^+_{\ell}(r)}, \;\;\; \; {}_{r_o}\tau_{\ell}^-(r) = \chi_{\ell}^-(r) \int^r_{r_o} \frac{dr'}{f(r') (\chi^-_{\ell}(r'))^2}.
\end{equation}
 The fact that, for $\Lambda \geq 0$, 
 $x$ ranges from minus infinity as $r \to {r_h}^+$, to infinity as $r \to \infty$ ($\Lambda=0$) or 
approaches the cosmological horizon ($\Lambda>0$), makes the algebraically special perturbations (\ref{chi+}), (\ref{tau+}) and 
(\ref{uns-}) uninteresting in these cases 
 because 
all these solutions  diverge at at least one of these 
two limits, and so are irrelevant as they are not eigenfunctions of $\mathcal{H}^{\pm}$. 
The situation is different for  $\Lambda <0$ and also for the non globally hyperbolic Schwarzschild naked singularity ($\Lambda=0, M<0$). 
For the latter,  
 $x$ can 
be chosen to range from $x=0$ (the timelike boundary at the $r=0$ singularity)
 to infinity (as $r \to \infty$) and,  as $w_{\ell}<0$ in this case, (\ref{chi+}) 
behaves properly in both limits (this happens only for even perturbation, 
neither $\chi^-(r)$ nor $\tau_{r_o}^-(r)$ for any $r_o$ behave properly).
 Moreover, it was found in  \cite{Gibbons:2004au} (see also \cite{Gleiser:2006yz})  that
 there is a single boundary condition at the $r=0$ timelike boundary that leads 
 to a consistent 
linear perturbation treatment;  therefore, the dynamics is not ambiguous in spite of the non globally hyperbolic character of the spacetime. 
This particular Robin boundary condition is precisely the one 
satisfied by the mode (\ref{chi+}). Since this mode  grows exponentially in time, 
 the claim that the Schwarzschild naked singularity is unstable is free of  ambiguities  \cite{Gleiser:2006yz}, \cite{Dotti:2008ta}. \\

In what follows we concentrate on the  $\Lambda<0$, $M>0$ SAdS${}_4$ black hole, for which $x \in (-\infty,0)$, 
 and one can  check using (\ref{x42}) that $\chi^+(r)$ and 
\begin{equation}\label{tau-}
\tau^-_{\ell} \equiv {}_{r_h}\tau_{\ell}^-(r) = \chi_{\ell}^-(r) \int^r_{r_h} \frac{dr'}{f(r') (\chi_{\ell}^-(r'))^2}
\end{equation}
 satisfy the bound state, negative energy requirements (\ref{ef}) and (\ref{bss}) for $\mathcal{H}^+_{(\ell,m)}$ 
and $\mathcal{H}^-_{(\ell,m)}$ respectively,  the choice $r_o=r_h$ above being 
crucial for this to hold. \\

The unstable even solution (\ref{un+}) 
  satisfies a Robin boundary condition (\ref{robin}) at $x=0$ with 
$\g$ equal to 
\begin{equation}
 \label{gg}
\g_{Ch} \equiv  w_{\ell} - \frac{2M\Lambda}{(\ell-1)(\ell+2)} = \frac{\Lambda r_h (\Lambda {r_h}^2-3)}{3(\ell-1)(\ell+2)}- 
\frac{(\ell+2)(\ell+1)\ell(\ell-1)}{2 r_h (\Lambda {r_h}^2-3)}
\end{equation}
The perturbation of the metric   is 
\begin{equation}\label{huns}
h^{+ \;unst}_{(\ell,m)} =  \exp({w_{\ell} v}) \; S_{(\ell,m)} \left[\frac{w_{\ell}}{6M} \left( r \ell(\ell+1)-6M \right) \; dv \otimes dv 
+ \frac{\ell (\ell+1)}{6M} r^2 \left(d\theta \otimes d\theta + \sin^2(\theta) \; d\phi \otimes d\phi \right)\right]
\end{equation}
where we defined $v=t+x$. The above expression shows that the perturbation 
 is well behaved across the future event horizon, defined by $r=r_h$, $v \in \mathbb{R}$. 
This  perturbation has the property of splitting only one  
of the two pairs of principal null directions of the background so that the perturbed spacetime 
is Petrov type-II. For generic perturbations, instead,  both pairs of principal null directions are split 
leaving a  type-I spacetime (see \cite{Araneda:2015gsa} for details).\\

For the pure mode (\ref{un+}), (\ref{G+}) gives 
\begin{equation} \label{egm}
G_+[\phi^{+ \;unst}_{(\ell,m)} ]= \frac{(\ell+2)!}{(\ell-2)!} \left(\frac{\ell (\ell+1) r-6M}{24 r^5}\right) \; \exp(w_{\ell} v), \;\;
 v= t+x(r).
\end{equation}

The existence of the bound state $\chi^+_{\ell}$ for the Zerilli SAdS${}_4$
Hamiltonian (\ref{dd22})-(\ref{hlm2}), which has a positive,  nonsingular potential (\ref{zp}),  
implies, in view of Proposition 5.ii,   that for even $\ell$ modes the critical 
value of 
$\g$ for instabilities satisfies 
\begin{equation}
\g_c \leq \g_{Ch}  \simeq \begin{cases}
\frac{(\ell+2)!}{(\ell-2)!} \frac{1}{6 r_h} &, r_h \to 0 \\[.5em]
\frac{\Lambda^2 {r_h}^3}{3 (\ell-1)(\ell+2)} &, r_h \to \infty,
\end{cases}
\end{equation}
 We can use this to test the upper bound (\ref{gcub}) in these limits 
using (\ref{vzi}):
\begin{equation} 
\g_c < \int_{-\infty}^0 V_{\ell}^{(+)}\; dx \simeq \begin{cases}  \frac{2\ell^2+2 \ell-3}{2 r_h} &, r_h \to 0 \\[.5em]
\frac{2 \Lambda^2 \; {r_h}^3}{ 3 (\ell+2)(\ell-1)} &, r_h \to \infty.
\end{cases}
\end{equation}
For large horizon radius $\g_{Ch}$ is half the value of the integrated potential, whereas for small $r_h$ we find that $\g_{Ch}$
is less than the integrated potential only for $\ell=2$ (the minimum possible $\ell$ value), and  grows as $\ell^4$ for large $\ell$, whereas 
the integrated potential grows only as $\ell^2$. 
We should keep in mind, however that the statement (\ref{gcub}) cannot be improved, as the $a \to 0$ limit of the step potential 
example in Section \ref{stepot} saturates this inequality.\\

\subsubsection{Boundary conditions and even/odd duality breaking}\label{natural}

Unlike the negative mass Schwarzschild solution, for which a unique  boundary condition at the conformal  timelike boundary is 
singled out from the infinite set of $z-$conditions ($z=$ D, N or $\g$)  by   a consistency requirement of the linear perturbation 
scheme 
\cite{Gibbons:2004au}, 
for SAdS${}_4$ we may choose among  
 Dirichlet, Neumann or Robin boundary conditions with a specific $\g_{(p=\pm,\ell,m)}$ for 
every mode. 
We recall that the potentials for even perturbations are nonsingular and positive definite, 
then Propositions 1 to 6 apply to them, whereas odd perturbations have  nonsingular potentials,  negative near the 
horizon for small enough $r_h/M$, then only Propositions 1-3 apply to them in general.\\
If, {\em for every} $(p=\pm,\ell,m)$  we choose the $\g_{(p=\pm,\ell,m)})$ below the
   critical value, the resulting dynamics will be stable.  However, a single  mode for which  $\g$ is high enough 
would introduce an instability. 
This implies that any gravitational stability claim for SAdS${}_4$ 
 is meaningful only after   specifying 
 the chosen boundary conditions. \\

In what follows, we proceed to analyze the relation between boundary conditions at the conformal boundary $r=\infty$ 
and the formal duality 
exchanging odd and even modes, discovered by Chandrasekhar about thirty years ago \cite{ch1} \cite{ch2}. 
This duality played a crucial role in the nonmodal stability proof for Schwarzschild black holes when $\Lambda \geq 0$ 
\cite{Dotti:2013uxa,Dotti:2016cqy} as it allows to replace all $\phi^+_{(\ell,m)}$ Zerilli fields by 
$\phi^-_{(\ell,m)}$ Regge-Wheeler field, and then use the boundedness and decay properties of the four dimensional
 Regge-Wheeler equation (\ref{4drwe}) for nonnegative $\Lambda$. This substitution generically fails when $\Lambda<0$, 
holding only when related boundary conditions are chosen in the odd and even sectors. These issues are explored in this Section.\\ 

All the relations we need follow from the observations in \cite{ch1} \cite{ch2} that 
\begin{equation} \label{key}
\mathcal{H}^{\pm}_{\ell}=\mathcal{D}^{\pm}_{\ell}\mathcal{D}^{\mp}_{\ell}-w^{2}_{\ell},
\end{equation}
where 
\begin{equation} \label{dpm}
\mathcal{D}^{\pm}_{\ell}=\pm\p_x+W_{\ell}, \;\;\;  W_{\ell} = w_{\ell}+\frac{6Mf}{r(r\m+6M)},
\end{equation}
and also that 
\begin{equation} \label{W}
W_{\ell} = \frac{\p_x \chi^+_{\ell}}{\chi^+_{\ell}} = - \frac{\p_x \chi^-_{\ell}}{\chi^-_{\ell}},
\end{equation}
which can be verified using (\ref{chi+}) and (\ref{uns-}).
From these we find that 
\begin{align} \label{vp1}
V^+_{\ell} &= W_{\ell}' + (W_{\ell})^2 -w^{2}_{\ell} = \frac{{\chi_{\ell}^+}''}{\chi_{\ell}^+} -w^{2}_{\ell}, \\
V^-_{\ell} &= -W_{\ell}' + (W_{\ell})^2 -w^{2}_{\ell} = \frac{{\chi_{\ell}^-}''}{\chi_{\ell}^-} -w^{2}_{\ell}
\end{align}
where a prime means derivative with respect to $x$. 
The second form in (\ref{vp1}) allows us to write the ordinary differential equation $\mathcal{H}^+_{\ell} \psi 
= E \psi$ as
\begin{equation}
 -\frac{\psi''}{\psi} + \frac{{\chi_{\ell}^+}''}{\chi_{\ell}^+} = (E+w^2_{\ell})
\end{equation}
from where, for  $E=-w^2_{\ell}$,  we readily obtain the two linearly independent solutions
$\psi=\chi^+_{\ell}$,  and $\psi=\tau^+_{\ell}$ given in (\ref{tau+}). 
A similar analysis leads  to the unstable odd modes (\ref{uns-}).\\

In view of (\ref{key}), acting with $\mathcal{D}^-_{\ell}$ on a solution to the differential equation 
$\mathcal{H}^+_{\ell} \psi^+ = E \psi^+$ gives a -possibly trivial- solution $\psi^- = \mathcal{D}_{\ell}^- \psi^+$ 
of $\mathcal{H}^-_{\ell} \psi^- = E \psi^-$ and viceversa:
\begin{align} \label{+2-}
\mathcal{H}^+_{\ell} \psi^+ = E \psi^+ &\Rightarrow  \mathcal{H}^-_{\ell} 
 (\mathcal{D}_{\ell}^- \psi^+) = E (\mathcal{D}_{\ell}^- \psi^+),\\ \label{-2+}
\mathcal{H}^-_{\ell} \psi^- = E \psi^- &\Rightarrow  \mathcal{H}^+_{\ell} 
 (\mathcal{D}_{\ell}^+ \psi^-) = E (\mathcal{D}_{\ell}^+ \psi^-).
\end{align}
From the equations above we find some trivial cases:
\begin{equation}
\mathcal{D}_{\ell}^- \chi_{\ell}^+ = 0, \;\;\;  \mathcal{D}_{\ell}^+ \chi_{\ell}^- = 0,
\end{equation}
Note that the most general solution of the equation $\mathcal{D}_{\ell}^- \psi^+ = 0$ ($\mathcal{D}_{\ell}^+ \psi^- = 0$) 
is a constant times $\chi_{\ell}^+$ ($\chi_{\ell}^-$)). For eigenfunctions with eigenvalues different from $-w_{\ell}^2$ 
 the effect 
of $\mathcal{D}_{\ell}^{\pm}$ can be reversed by $\mathcal{D}_{\ell}^{\mp}$ (times a function of $E$), and viceversa:
\begin{align} \label{inversed}
\mathcal{H}^+_{\ell} \psi^+ = E \psi^+ &\Rightarrow  \mathcal{D}_{\ell}^+ (\mathcal{D}_{\ell}^- \psi^+)
= (\mathcal{H}^{+}_{\ell}+ w_{\ell}^2) \psi^+ = (E+w_{\ell}^2) \psi^+ \neq 0 \\
\mathcal{H}^-_{\ell} \psi^- = E \psi^- &\Rightarrow  \mathcal{D}_{\ell}^- (\mathcal{D}_{\ell}^+ \psi^-)  
= (\mathcal{H}^{-}_{\ell}+ w_{\ell}^2) \psi^- = (E+w_{\ell}^2) \psi^- \neq 0. \label{inversed2}
\end{align}
From equations (\ref{inversed}) (\ref{inversed2}) follows that if $\psi^-_j, j=1,2,$ are two linearly independent solutions 
of $\mathcal{H}^-_{\ell} \psi_j^- = E \psi^-_j$ with $E \neq -w_{\ell}^2$, then 
$\psi_j^+ = \mathcal{D}_{\ell}^+ \psi^-_j, j=1,2, $ are two linearly independent solutions of 
$\mathcal{H}^+_{\ell} \psi_j^+ = E \psi^+_j$,  and similarly if we exchange $-$ and $+$.\\

We also note from the above equations  that
\begin{equation}\label{trivial2}
 \mathcal{D}_{\ell}^- {}_{r_o}\tau_{\ell}^+ = \chi^-_{\ell}, \;\;\;  \mathcal{D}_{\ell}^+ {}_{r_o}\tau_{\ell}^- = \chi^+_{\ell}
\end{equation}
for any $r_o$, and that
\begin{equation} \label{kappa}
\mathcal{D}_{\ell}^+ \kappa^-_{\ell} = \chi^+_{\ell} \Rightarrow 
 \kappa^-_{\ell} = \tau_{\ell}^- + \alpha \chi^-_{\ell},
\end{equation}
where $\tau_{\ell}^-$ was defined in (\ref{tau-}) and $\a$ is a constant.\\

The possibility of exchanging even and odd modes using $\mathcal{D}_{\ell}^{\pm}$ is the duality, peculiar to four dimensions, 
that we  
will analyze for $\Lambda<0$ in the remaining  of this Section.  
For $\Lambda \geq 0$, the fields $\phi^{\pm}_{(\ell,m)}$ belong to $L^2(\mathbb{R},dx)$ and 
the operators $\mathcal{D}^{\pm}_{\ell}$ give a bijection between the sets of solutions of the odd and even 1+1 wave equations 
(see Section 4.5 in \cite{Dotti:2016cqy}.) 
The case where $\Lambda<0$ is much subtler. The 
  linear gravity potentials $V_{\ell}^{\pm}$ 
are nonsingular, the values at $x=0$  of solutions of $\mathcal{H}^{\pm}_{\ell} \psi^{\pm} = E \psi^{\pm}$ and their 
$x-$derivatives are well defined and generically non-zero (equation (\ref{lst1})), so Dirichlet, Neumann or Robin boundary conditions are allowed. Suppose that, for a 
a given $(\ell,m)$, we choose 
\begin{equation} \label{zplus}
{\psi^+}' \big|_{x=0} = \g_e \; \psi^+\big|_{x=0} 
\end{equation}
where, in what follows $\g_e \in \mathbb{R} \cup \{ \infty\}$ to include the cases  $\g_e=0$ (Neumann) and $\g_e=\infty$ (Dirichlet), and we similarly introduce $\g_o$ for the odd modes, dropping the $(\ell,m)$ indices for simplicity. \\

From (\ref{dpm}) we find that 
\begin{equation}
(\mathcal{D}_{\ell}^- \psi^+)\big|_{x=0} = -{\psi^+}' \big|_{x=0} + (W_{\ell}\;  {\psi^+}) \big|_{x=0} 
= (W_{\ell}-\g_e) {\psi^+} \big|_{x=0}
\end{equation}
and that 
\begin{equation} \label{d2}
(\mathcal{D}_{\ell}^- \psi^+)'\big|_{x=0} = -{\psi^+}'' \big|_{x=0} + (W_{\ell}'+ \g_e W_{\ell}) {\psi^+} \big|_{x=0},
\end{equation}
where, from (\ref{W})
\begin{equation} \label{d1}
W\big|_{x=0} = \g_{Ch} \;\;\;\text{ and } \;\;\;  W'\big|_{x=0} =  \left( \frac{2M\Lambda}{\mu}\right)^2,
\end{equation}
and  $\g_{Ch}$ was defined in (\ref{gg}). 
It is easy to prove  from these two equations that, in general, 
there is no function $\g_o(\g_e)$, $\g_e, \g_o \in \mathbb{R} \cup \{\infty \}$,  such that ${\psi^+}'/\psi^+=\g_e$ at $x=0$  implies 
 $(\mathcal{D}_{\ell}^- \psi^+)'/\mathcal{D}_{\ell}^- \psi^+=\g_o(\g_e)$ at $x=0$. 
 To show this, we use the fact that  $\psi^+$ fields satisfying  (\ref{zplus}) can be expanded 
using the complete basis of generalized eigenfunctions ${}^{\g_e}\psi^+_E$  (we suppress the $\ell$ index) 
of the corresponding self adjoint extension ${}^{\g_e}\mathcal{H}^+_{\ell}$, so that  
\begin{equation}\label{psie}
\psi^+ = \int dE \, c_E \, {}^{\g_e}\psi^+_E,
\end{equation}
where the integral notation includes a sum over bound states, if there were any.\\

For  an energy eigenstate we find from  (\ref{d2}) that 
\begin{align}
\left. (\mathcal{D}_{\ell}^- \;\; {}^{\g_e}\psi^+_E)' \right|_{x=0} &= 
(\mathcal{H}^+_{\ell} - V_{\ell}^++  W_{\ell}'+ \g_e W_{\ell})\; 
 {{}^{\g_e}\psi^+_E} \big|_{x=0}\\
&= (E+w_{\ell}^2 - W_{\ell}^2+ \g_e W_{\ell}) \; {{}^{\g_e}\psi^+_E} \big|_{x=0},
\end{align}
which, together with  (\ref{zplus}) gives 
\begin{equation}\label{quot}
\left. \frac{(\mathcal{D}_{\ell}^-\;  {}^{\g_e}\psi^+_E)'}{\mathcal{D}_{\ell}^-\;  {}^{\g_e}\psi^+_E} \right|_{x=0}= 
\left. \frac{E+w_{\ell}^2 - W^2_{\ell} + \g_e W_{\ell}}{W_{\ell}-\g_e} \right|_{x=0}.
\end{equation}
Since, generically, the quotient above depends on $E$,  ${\psi^+}'/\psi^+$ in (\ref{psie}) 
will  have different values for different functions in the linear space obtained by applying $\mathcal{D}^-_{\ell}$ 
to the domain of 
${}^{\g_e}\mathcal{H}^+_{\ell}$,
 then  the dynamics will not be defined in this  space since  it is not a 
self adjoint domain 
of $\mathcal{H}^-_{\ell}$ (Note that this problem cannot be fixed by the alternative definitions  
$\tilde{\mathcal{D}}^- \equiv 
f(E) \mathcal{D}^-$ of the the dual map used  in \cite{ch2} and \cite{Bakas}). The only exceptions (i.e., situations where 
the right hand side of (\ref{quot}) does not depend on $E$)  are: i)  when we choose 
Dirichlet boundary conditions in the even sector, that is 
$\g_e \to \infty$ in (\ref{quot}), which gives Robin conditions in the odd sector with 
$\g_{o} = -W\big|_{x=0}=-\g_{Ch}$, that is 
\begin{equation} \label{o1}
\g_{o} = -\g_{Ch} \;\;\; (\g_{e} = \infty),
\end{equation}
and ii) when we choose Dirichlet boundary conditions in the odd sector, that is, $\g_e=W|_{x=0}=\g_{Ch}
$ in (\ref{quot})):
\begin{equation}
\g_{e} = \g_{Ch} \;\;\; (\g_{o} = \infty).
\end{equation}
The following proposition gives more details about   the supersymmetry and  these two cases:\\

\begin{prop}\label{susy}
Consider the maps $\mathcal{D}^{\pm}$ defined in (\ref{dpm}). In what follows we use the symbol 
${}^{\g}\mathcal{H}^{\pm}_{\ell}$ 
both for the self adjoint operator and its domain. 
\begin{itemize}
\item[i)] The spectra of $ {}^{D}\mathcal{H}^+_{\ell}$ and ${}^{-\g_{Ch}}\mathcal{H}^-_{\ell}$  are nonnegative. \\
$\mathcal{D}^-_{\ell}: {}^{D}\mathcal{H}^+_{\ell} \to {}^{-\g_{Ch}}\mathcal{H}^-_{\ell}$ 
is a bijection.  
\item[ii)] The spectrum of ${}^{D}\mathcal{H}^-_{\ell}$ is nonnegative, that 
of ${}^{\g_{Ch}}\mathcal{H}^+_{\ell}$ contains a  negative energy with eigenfunction $\chi^+$. \\
The map   $\mathcal{D}^-_{\ell}: {}^{\g_{Ch}}\mathcal{H}^+_{\ell} \to {}^{D}\mathcal{H}^-_{\ell}$ is surjective 
and has kernel the linear space generated by  $\chi^+_{\ell}$.\\ The map 
 $\mathcal{D}^+_{\ell}: {}^{D}\mathcal{H}^-_{\ell} \to {}^{\g_{Ch}}\mathcal{H}^+_{\ell}$ is injective. 

\item[iii)] There are no  other values of $\g, \g' \in \mathbb{R} \cup \{\infty \}$ 
such that $\mathcal{D}^{\mp}_{\ell}({}^{\g}\mathcal{H}^{\pm}_{\ell}) \subset {}^{\g'}\mathcal{H}^{\mp}_{\ell}$ 
\end{itemize}
\end{prop}
\begin{proof}

We have already proven iii).\\

 To prove i) note that the only solution of $\mathcal{D}^-_{\ell} \psi^+ =0$ is a constant times 
$\chi^+$, which does not belong to   ${}^{D}\mathcal{H}^+_{\ell}$, so the map in 
i) has a trivial kernel and therefore is injective. Similarly, the only solution of $\mathcal{D}^+_{\ell} \psi^- =0$ 
is a constant times $\chi^-=1/\chi^+$ which, although  satisfies a Robin condition with $\g=-\g_{Ch}$ at $x=0$, 
diverges as $x \to -\infty$ and so does not belong to ${}^{-\g_{Ch}}\mathcal{H}^-_{\ell}$ . This implies that 
the map $\mathcal{D}^+_{\ell}: {}^{-\g_{Ch}}\mathcal{H}^-_{\ell} \to {}^{D}\mathcal{H}^+_{\ell}$ is injective. 
Since $V^+_{\ell}$ is nonnegative, Proposition \ref{stable} applies  and the spectrum of  
${}^{D}\mathcal{H}^+_{\ell}$ is nonnegative. The spectrum of ${}^{-\g_{Ch}}\mathcal{H}^-_{\ell}$ must also be nonnegative, 
otherwise, a negative eigenfunction of ${}^{-\g_{Ch}}\mathcal{H}^-_{\ell}$ would be sent by the injective map $\mathcal{D}_{\ell}^+$ to a negative 
eigenfunction of ${}^{D}\mathcal{H}^+_{\ell}$, which is a contradiction. This proves that both 
${}^{D}\mathcal{H}^+_{\ell}$ and ${}^{-\g_{Ch}}\mathcal{H}^-_{\ell}$ have  nonnegative spectra 
 and therefore  lead to stable 
dynamics in the even and odd sectors respectively. 
To prove that $\mathcal{D}^-_{\ell}: {}^{D}\mathcal{H}^+_{\ell} \to {}^{-\g_{Ch}}\mathcal{H}^-_{\ell}$  
is onto we proceed as in  Lemma 7 in \cite{Dotti:2016cqy}: 
let $\psi^-$ be an arbitrary function in  ${}^{-\g_{Ch}}\mathcal{H}^-_{\ell}$ and 
$\psi^-= \int dE \; q(E) \; {}^{-\g_{Ch}}\psi_E^-$  its expansion in eigenfunctions ${}^{-\g_{Ch}}\psi_E^-$ 
of ${}^{-\g_{Ch}}\mathcal{H}^-_{\ell}$. 
 Since the spectrum of  ${}^{-\g_{Ch}}\mathcal{H}^-_{\ell}$  is nonnegative,  
\begin{equation}
\tilde \psi^-= \int dE \; \left(\frac{q(E)}{w_{\ell}^2+ E}\right) \; {}^{-\g_{Ch}}\psi_E^-
\end{equation}
 is well defined and belongs to ${}^{-\g_{Ch}}\mathcal{H}^-_{\ell}$. Then 
$\mathcal{D}_{\ell}^+\tilde \psi^-$ is in ${}^{D}\mathcal{H}^+_{\ell}$
and is the function sent to $\psi^-$ by 
$\mathcal{D}_{\ell}^-$, as the following 
calculation shows:
\begin{equation}
\mathcal{D}_{\ell}^- (\mathcal{D}_{\ell}^+\tilde \psi^-) = (\mathcal{H}_{\ell}^- + w_{\ell}^2) 
\int dE \; \left(\frac{q(E)}{w_{\ell}^2+ E}\right) \; {}^{-\g_{Ch}}\psi_E^- =\psi^-
\end{equation} 
This completes the proof of i).\\

 To prove ii)  recall that the only solution of $\mathcal{D}_{\ell}^+ \psi^-=0$ is 
a constant times $\chi^-$, which does not belong to    ${}^{D}\mathcal{H}^-_{\ell}$, therefore 
$\mathcal{D}^+_{\ell}: {}^{D}\mathcal{H}^-_{\ell} \to {}^{\g_{Ch}}\mathcal{H}^+_{\ell}$ is injective. 
We have already proven that the kernel of $\mathcal{D}^-_{\ell}: {}^{\g_{Ch}}\mathcal{H}^+_{\ell} \to {}^{D}\mathcal{H}^-_{\ell}$ is 
the one dimensional space of functions proportional to  $\chi^+_{\ell}$,  which is the eigenfunction of the only (Proposition 4.iii) 
negative energy in the spectrum of ${}^{\g_{Ch}}\mathcal{H}^+_{\ell}$. If  ${}^{D}\mathcal{H}^-_{\ell}$  admitted 
a negative energy, an eigenfunction $\kappa^-_{\ell}$ of this energy would be sent to a negative energy eigenfunction of ${}^{\g_{Ch}}\mathcal{H}^+_{\ell}$ 
by the injective map $\mathcal{D}^+_{\ell}$; i.e., we may assume that $\mathcal{D}^+_{\ell} \kappa^-_{\ell} 
= \chi^+_{\ell}$. 
However, it follows (\ref{kappa}) that the general solution of the differential equation 
$\mathcal{D}^+_{\ell} \kappa = \chi^+_{\ell}$ 
is  $\tau^-_{\ell}+\alpha \chi^-_{\ell}$ and, given that $\tau^-_{\ell}(x=0)>0$, we need $\alpha \neq 0$ 
for $\tau^-_{\ell}+\alpha \chi^-_{\ell}$ to equal zero at $x=0$ and, since $\a \neq 0$, 
 the resulting function diverges as $x \to -\infty$ 
and therefore does not belong to ${}^{D}\mathcal{H}^-_{\ell}$, so we reach a contradiction and conclude that 
${}^{D}\mathcal{H}^-_{\ell}$ has a nonnegative spectrum. 
This allows us to prove that $\mathcal{D}^-_{\ell}: {}^{\g_{Ch}}\mathcal{H}^+_{\ell} \to {}^{D}\mathcal{H}^-_{\ell}$ is surjective 
proceeding as above:
expand ${}^{D}\mathcal{H}^-_{\ell} \ni \psi^- = \int dE \, s(E) \, {}^D\psi_E^-$ in a basis of generalized eigenfunctions ${}^D\psi_E^-$ of 
the positive definite operator ${}^{D}\mathcal{H}^-_{\ell}$. Consider the function 
\begin{equation}
\hat \psi^-= \int dE \; \left(\frac{s(E)}{w_{\ell}^2+ E}\right) \; {}^D\psi_E^-.
\end{equation}
Note that $\mathcal{D}^+_{\ell}\hat \psi^-$ belongs to the domain of ${}^{\g_{Ch}}\mathcal{H}^+_{\ell}$ and that 
 $\mathcal{D}^-_{\ell}$ sends it to $\psi^-$.
\end{proof}

\section{Summary}\label{discussion}

In what follows we enumerate the subjects addressed (items 1 and 2 below) and the results proven (items 3 and 4) 
in this work:\\
\begin{center}
{\em 1. Stability of scalar field as an indicator of gravitational stability}
\end{center}
The stability of a scalar test field on a given spacetime is oftentimes taken  as indicative  of linear gravitational stability. 
SAdS${}_4$ offers  an example of how naive this idea can be: although there is a single choice of boundary condition for a scalar 
field on SAdS${}_4$, under which the   field is stable, there are infinitely many possible dynamics for gravitational perturbations. 
If any  mixture of  Dirichlet, Neumann or Robin boundary conditions with $\g$
below the critical value is chosen for the different modes, the evolution of gravitational perturbations will be stable. 
If, on the contrary, a Robin boundary condition with $\g$ above the critical value is allowed {\em for a single}
 mode , the perturbation will be 
unstable, the instability being signaled by an exponentially growing mode similar to (\ref{huns}). 
Similar comments apply to Maxwell fields on SAdS${}_4$.\\

\begin{center}
{\em 2. Naturalness of Robin boundary conditions}
\end{center}
Robin boundary conditions on the 1+1 auxiliary fields satisfying (\ref{dd22})-(\ref{hlm2})  are 
enforced by the self-consistence of the linear perturbation treatment in nakedly singular spacetimes 
(see references \cite{Gibbons:2004au} and \cite{Gleiser:2006yz}). They are also  
 a natural choice 
when studying gravitational perturbations of asymptotically AdS spacetimes, as they result 
from the imposition of Dirichlet or Neumann conditions on  geometrically meaningful  fields, as discussed in 
Section \ref{natural}. Robin  conditions also arise in the context of AdS-CFT dualities, 
 and if we want to preserve 
the even/odd  duality in SAdS${}_4$ (see item 4 below).\\

\begin{center}
{\em 3. Robin instabilities in 1+1 wave equations}
\end{center} For any physical system reducing to equation  (\ref{2dwe}) on the $x<0$ half space with  a 
nonsingular potential,  there are instabilities for high enough Robin parameter $\g$. If the potential 
is nonnegative, there is a critical value $\g_c>0$ such that the set of unstable boundary conditions 
is of the form $\g > \g_c$. A number of properties about 
the energy spectrum  of the associated quantum Hamiltonian on a half line and its  bound state are 
given in sections \ref{grt1p} and \ref{nn} (Propositions 1-6). 
The mechanism triggering the instabilities of (\ref{2dwe}) is explained  within 
Section \ref{ec} and illustrated using simple toy models in Section \ref{toy}. \\

 \begin{center}
{\em 4. Even/odd duality and stability of  SAdS${}_4$}
\end{center}

Four dimensional Schwarzschild black holes  exhibit
 a unique feature of a duality exchanging even and odd modes which is due to 
the fact that the corresponding potentials form a supersymmetric pair. This is used 
in \cite{Dotti:2016cqy}  to extend to the even sector the proof of 
nonmodal 
stability  for Schwarzschild black holes when $\Lambda \geq 0$. In the asymptotically AdS case, however, 
the even/odd duality is obstructed by the  boundary conditions at the timelike boundary. 
We have found that 
there are only two boundary conditions compatible with  the even/odd symmetry, Dirichlet in the even sector and  
Robin with a  particular, $(\ell,m)-$dependent, $\g$ in the odd one, and viceversa, 
with only the first one leading  to a stable dynamics (Proposition 7 in Section \ref{robin4}).\\
 An explicit unstable 
even gravitational mode for a specific $\g$ is shown in this section together with its effect on the background geometry.

\section{Acknowledgments}
This work was partially funded by grants PIP 11220080102479
(Conicet-Argentina) and 30720110101569CB (Universidad Nacional de C\'ordoba). 
B.A. is supported by a fellowship from Conicet.

\end{document}